\documentclass[11pt,english]{article}
\usepackage{lmodern}
\usepackage{color}

\usepackage[T1]{fontenc}
\usepackage[utf8x]{inputenc}
\usepackage{amsthm}
\usepackage{amsmath}
\usepackage{amssymb}
\usepackage{esint}
\usepackage{textcomp}
\usepackage{graphicx}
\usepackage{float}

\makeatletter
  \newtheorem{rem}{\protect\remarkname}
  \theoremstyle{plain}
  \newtheorem{lem}{\protect\lemmaname}
  \theoremstyle{plain}
  \newtheorem{prop}{\protect\propositionname}
  \theoremstyle{plain}
  \newtheorem{definition}{\protect Definition}

\@ifundefined{date}{}{\date{}}
\usepackage{a4wide}

\def\logexpo{\frac{1}{\gamma}\log\Biggl(\mathbb{E}\Biggl[\exp\Biggl(\gamma\Biggl(}
\def\logexpc{\Biggr)\Biggr)\Biggr]\Biggr)}

\def \E{\mathbb{E}}
\def \P{\mathbb{P}}
\def \R{\mathbb{R}}

\def \Sum{\displaystyle\sum}

\DeclareMathOperator*{\argmin}{arg\,min}

\title{Accelerated Share Repurchase: pricing and execution strategy\thanks{This research has been conducted with the support of the Research Initiative ``Ex\'ecution optimale et statistiques de la liquidit\'e haute
fr\'equence'' under the aegis of the Europlace Institute of Finance. The
authors would like to thank Robert Almgren (Quantitative Brokers), Nicolas Grandchamp des Raux (HSBC France), Charles-Albert Lehalle (CFM), Terry Lyons (Oxford Man), Ramzi Maghrebi (HSBC France), Huyen Pham (Universit\'e Paris-Diderot), Chris Rogers (Cambridge) and Mathieu Rosenbaum (UPMC) for the discussions we had on the topic.}}
 \author{Olivier {\sc Gu\'eant} \footnote{Universit\'e Paris-Diderot, UFR de
Math\'ematiques, Laboratoire Jacques-Louis Lions, gueant@ljll.univ-paris-diderot.fr}
 \and Jiang {\sc Pu}  \footnote{Institut Europlace de Finance. Research Initiative ``Ex\'ecution optimale et statistiques de la liquidit\'e haute
fr\'equence'', jiang.pu.2009@m4x.org }
 \and Guillaume {\sc Royer}\footnote{CMAP, Ecole Polytechnique Paris, guillaume.royer@polytechnique.edu.}}

\makeatother

\usepackage{babel}
  \providecommand{\lemmaname}{Lemma}
  \providecommand{\propositionname}{Proposition}
  \providecommand{\remarkname}{Remark}

\begin{document}
\maketitle
\begin{abstract}
In this article, we consider the optimal execution problem
associated to accelerated share repurchase contracts. When firms want
to repurchase their own shares, they often enter such a contract with
a bank. The bank buys the shares for the firm and is paid the average
market price over the execution period, the length of the period being
decided upon by the bank during the buying process. Mathematically,
the problem is new and related to both option pricing (Asian and Bermudan
options) and optimal execution. We provide a model, along with associated
numerical methods, to determine the optimal stopping time and the
optimal buying strategy of the bank.

\vspace{10mm}

\noindent \textbf{Key words:} Optimal execution, ASR contracts, Optimal stopping,
Stochastic optimal control, Utility indifference pricing. \vspace{5mm}

\end{abstract}

\section{Introduction}

The mathematical literature on optimal execution often deals with
the trade-off between execution costs and price risk. Price risk may
be measured with respect to different benchmark prices: arrival price
in the case of an Implementation Shortfall (IS) order, closing price
of the day in the case of a Target Close (TC) order or VWAP over a
given period of time in the case of a VWAP order. In all these cases,
the definition of the benchmark (although not its value) is known
ex-ante. It is common however, for very large orders, especially when
a firm wants to (re)purchase its own shares, to consider a benchmark
price that is an average price over a period that is decided upon
over the course of the execution process. More precisely, in the case
of an Accelerated Share Repurchase contract, or more exactly a post-paid
Accelerated Share Repurchase (hereafter ASR) with fixed nominal, the bank has to deliver
the firm $Q$ shares at an exercise date $\tau$ chosen by the
bank over the course of the execution process among a set of pre-defined
dates. Then, in exchange for the shares, the firm pays the bank, at
date $\tau$, the arithmetic average over the period $[0,\tau]$ of the daily
VWAP prices.

An ASR is therefore an Asian-type option with Bermuda-style exercise
dates. Moreover, because quantities to deliver are usually large, we
need to take account of market impact and execution costs. The problem
we consider is therefore both a problem of option pricing and a problem
of optimal execution.

Option pricing and hedging with execution costs have been considered
mainly by Rogers and Singh \cite{rogerssingh} and then by Li and
Almgren in \cite{lialmgren}. In their settings, as opposed to the
literature on transaction costs (see for instance \cite{tcbarlessoner,tccvitanic1,tccvitanic2,tcleland})
and in line with the literature on optimal liquidation (see the seminal
paper by Almgren and Chriss \cite{almgren} and the two recent papers
\cite{gueant,schied}), the authors consider execution costs that
are not linear in (proportional to) the volume executed but rather strictly
convex to account for liquidity effects. Rogers and Singh consider
an objective function that penalizes both execution costs and mean-squared
hedging error at maturity. They obtain, in this close-to-mean-variance
framework, a closed form approximation for the optimal hedging strategy
when illiquidity costs are small. Li and Almgren, motivated by saw-tooth
patterns recently observed on five US stocks (see \cite{lehalle1,lehalle2}),
considered a model with both permanent and temporary impact. In their
model, they assume execution costs are quadratic and they consider a constant
$\Gamma$ approximation. Using a different objective function, they obtain a closed
form expression for the hedging strategy. Both papers do not consider
physical settlement but rather cash settlement%
\footnote{Recently, Guéant and Pu \cite{gueantpu} also proposed a method to
price and hedge a vanilla option in a utility-based framework, under
general assumptions on market impact. In particular they considered physical
settlement.%
} and ignore therefore part of the costs. Moreover, they only consider
European payoffs and our problem is therefore more complex.

In a recent working paper, Jaimungal et al. \cite{jaimungal} proposed
a model for Accelerated Share Repurchase contracts. Interestingly, they managed
to reduce the problem to a 3-variable PDE whereas the initial problem
is in dimension 5. The PDE they obtained is then solved numerically.
Our model differs from their model in many ways. A first important
difference has to do with market impact modeling: in \cite{jaimungal}, the authors
only considered the case of quadratic execution costs while we present
a model with general assumptions for market impact functions, with
both temporary market impact and permanent market impact. Another
important point is that their model considered the case of a risk-neutral
agent with inventory penalties whereas our model is more general in
that we consider a risk-averse agent in a utility-based framework,
allowing then indifference pricing. In the framework we propose, we
also manage to reduce the 5-variable Bellman equation associated to
the problem to a 3-variable equation that can be solved numerically
using classical tools. Our model also differs from the model presented
in \cite{jaimungal} as the authors used a continuous model and replaced
the discrete average of daily VWAP prices by a global TWAP. As a consequence of that, we avoid the computational difficulties induced by continuous time (see the numerical issues mentioned in \cite{jaimungal}). In our discrete time model, we consider, in line with
the definition of the payoff, daily fixing of prices and a payoff linked to
the arithmetic average of daily VWAP prices since inception. Finally, a minor difference is that
they restrict their model to a variant of ASR where the bank can exercise
at any time, whereas we consider the more realistic case of Bermudan
exercise dates.

In Section \ref{sect: setup}, we present the basic hypotheses of our model and the
Bellman equation in dimension 5. In Section \ref{sect:reduc-var}, we show how to go from
the Bellman equation with 5 variables to a 3-variable equation. In Section
\ref{sect: permanent market impact}, we introduce permanent market impact. In Section \ref{sect: numerics}, we develop
a numerical method to compute the solution of our problem using trees of polynomial size,
and we present numerical examples.

\section{Setup of the model}
\label{sect: setup}
The problem we consider is the problem of a bank which is asked by
a firm to repurchase $Q>0$ of the firm's own shares through an accelerated
share repurchase (ASR) contract of maturity $T$. By definition of
the contract, the bank will buy the shares of the firm over a sub-period
of $[0,T]$ and can deliver them at contractually specified dates.
At time of delivery, the firm receives the shares and pays the average
of the daily VWAPs since inception for each of the $Q$ shares.

\subsection{Problem formulation}

We consider a discrete model where each period of time corresponds
to one day (of length $\delta t$). In other words, if the interval $[0,T]$ corresponds to $N$ days ($T= N\delta t$), we consider the subdivision $(t_n)_{0\le n \le N}$ with $t_n = n \delta t$ and we index variables with $n$ for the sake of simplicity.

In order to introduce random variables, we consider a filtered probability
space $\left(\Omega,\left(\mathcal{F}_{n}\right)_{0\le n\le N},\P\right)$
satisfying the usual assumptions.

To model prices, we start with an initial price $S_{0}$ and we consider
that the dynamics of the price (in absence of market impact) is given
by:
$$
\forall n\in\lbrace0,\ldots,N-1\rbrace,\quad S_{n+1}=S_{n}+\sigma\sqrt{\delta t}\epsilon_{n+1},
$$
where $\epsilon_{n+1}$ is $\mathcal{F}_{n+1}$-measurable and where
$\left(\epsilon_{n}\right)_{n}$ are assumed to be i.i.d. with
mean $0$, variance $1$, and a moment-generating function defined on $\R_+$.

To stick to the definition of the payoff of an Accelerated Share Repurchase
contract, we consider that, for $n\ge 1$, $S_{n}$ is the daily VWAP over the
period $[t_{n-1},t_n]$, this VWAP being known at time $t_n$.\footnote{We could also consider that $S_n$ is the closing price of day $n$ and then approximate the payoff of an ASR using the average of daily closing prices.}

We also introduce the process $\left(A_{n}\right){}_{n\ge1}$ that
stands for the arithmetic average of daily VWAPs over the period $[0,t_n]$:
$$
A_{n} =\frac{1}{n}\sum_{k=1}^{n}S_{k}.
$$
To buy shares, we assume that the bank sends every day an order to
be executed over the day. The size of the order
sent by the bank over $[t_{n},t_{n+1}]$ is denoted $v_{n}\delta t$ where the process $v$
is assumed to be adapted. Hence, the number of shares that remain
to be bought at time $t_n$ (denoted $q_{n}$) verifies:
$$
\left\{ \begin{array}{lcl}
q_{0} & = & Q\\
q_{n+1} & = & q_{n}-v_{n}\delta t, \quad \forall n \le N-1
\end{array}\right.
$$

The price paid by the bank for the shares bought over $[t_n,t_{n+1}]$
is assumed to be the VWAP price $S_{n+1}$ plus execution costs. To model execution costs, we introduce a function $L\in C(\R,\R_{+})$
verifying:%

\begin{itemize}
\item $L(0)=0$,
\item $L$ is an even function,
\item $L$ is increasing on $\R_{+}$,
\item $L$ is strictly convex,
\item $L$ is asymptotically superlinear, that is:
\begin{eqnarray*}
\lim_{\rho\to+\infty}\frac{L(\rho)}{\rho} & = & +\infty.
\end{eqnarray*}

\end{itemize}
The cash spent by the bank is modeled by the process $\left(X_{n}\right)_{0\le n\le N}$
defined by:
$$
\left\{ \begin{array}{lcl}
X_{0} & = & 0\\
X_{n+1} & = & X_{n}+v_{n}S_{n+1}\delta t+L\left(\dfrac{v_{n}}{V_{n+1}}\right)V_{n+1}\delta t
\end{array}\right.
$$
where $V_{n+1}\delta t$ is the market volume over the period $[t_n,t_{n+1}]$.
We assume that the process $(V_{n})_{n}$ is $\mathcal{F}_{0}$-measurable.

\begin{rem}
In applications, $L$ is often a power function, i.e. $L(\rho)=\eta\left|\rho\right|^{1+\phi}$
with $\phi>0$, or a function of the form $L(\rho)=\eta\left|\rho\right|^{1+\phi}+\psi|\rho|$
with $\phi,\psi>0$. In other words, the execution costs per share are of the form $\eta |\rho|^\phi + \psi$ where $\rho$ is the participation rate. This form of temporary market impact function is supported by the empirical study made by Almgren et al. in \cite{almgrenciti}.
\end{rem}

We then introduce a non-empty set $\mathcal{N}\subset\lbrace1,\ldots,N-1\rbrace$
corresponding to time indices at which the bank can choose whether
it delivers the shares to the firm (note that $N\notin\mathcal{N}$
because the bank has no choice at time $t_N = T$: it is forced to deliver
at terminal time $T$ if it has not done so beforehand). In practice, this set
is usually of the form $\{n_{0},\ldots,N-1\}$ where $n_{0}>1$. At
time $n^{\star}\in\mathcal{N}\cup\{N\}$ of delivery, we assume that
the shares remaining to be bought ($q_{n^{\star}}$) can be purchased
for an amount of cash equal to $q_{n^{\star}}S_{n^{\star}}+\ell(q_{n^{\star}})$
where $\ell$ is a penalty function, assumed to be convex, even, increasing
on $\R_{+}$ and verifying $\ell(0)=0$.%
\footnote{One can set $\ell$ high enough outside of $0$ to prevent delivery
whenever $q$ is different from $0$. Considering the convex indicator $\ell(q) = +\infty 1_{q\neq0}$ is also possible.}

The optimization problem the bank faces is therefore:
\begin{equation}
\sup_{(v,n^{\star})\in\mathcal{A}}\E\left[-\exp\left(-\gamma\left(QA_{n^{\star}}-X_{n^{\star}}-q_{n^{\star}}S_{n^{\star}}-\ell(q_{n^{\star}})\right)\right)\right],\label{def maximisation}
\end{equation}
where $\mathcal{A}$ is the set of admissible strategies defined as:
$$
\begin{array}{cccl}
\mathcal{A} & = & \Big\{(v,n^{\star})\ \Big| & v=(v_{n})_{0\le k\le n^{\star}-1}\text{ is }\left(\mathcal{F}\right)\text{-adapted},\\
 &  &  & n^{\star}\text{ is a }\left(\mathcal{F}\right)\text{-stopping time taking values in }\mathcal{N}\cup\{N\}\Big\}
\end{array}
$$
and where $\gamma$ is the absolute risk aversion parameter of the
bank.

Solving this problem requires to determine both the optimal execution
strategy of the bank and the optimal stopping time.

We use this utility maximization setup to define the indifference price of the ASR contract.

\begin{definition}[Price of the ASR contract]
We define the indifference price $ \Pi$ of the ASR contract as:
$$ \Pi:=\inf \Big\{ p \ \Big| \ \sup_{(v,n^{\star})\in\mathcal{A}}\E\left[-\exp\left(-\gamma\left(p+QA_{n^{\star}}-X_{n^{\star}}-q_{n^{\star}}S_{n^{\star}}-\ell(q_{n^{\star}})\right)\right)\right]\geq -1   \Big\} $$
\end{definition}

This is exactly the opposite value of the certainty equivalent associated to writing an ASR contract.

\begin{rem}
\label{oboidorman}
For the bank, the main interest of the ASR contract is its optionality
component. The bank can deliver the shares at a date it chooses, and
can hence benefit from the difference between the price of execution
and the average price since inception. For that reason, and if there
were no execution costs, the price $\Pi$ of an ASR would be negative.
An important question in practice is to know whether the gain associated
to the optionality component compensates the liquidity costs and the
risk of the contract, that is whether $\Pi$ is positive or negative.
\end{rem}

\subsection{The value function and the Bellman equation}

In this section, we introduce the value function associated to the initial problem \eqref{def maximisation} and we characterize its evolution through the associated Bellman equation. Namely, we define:%
\footnote{We will see below that $u_{0}$ does not depend on $A$. This is coherent
with the fact that $A_{n}$ is defined only for $n\ge1$.%
}
\begin{align}\label{def maximisation dynamique}
u_{n}(x,q,S,A)=  \sup_{(v,n^{\star})\in\mathcal{A}_{n}}\E\left[-\exp\left(-\gamma\left(QA_{n^{\star}}^{n,A,S}-X_{n^{\star}}^{n,x,v}-q_{n^{\star}}^{n,q,v}S_{n^{\star}}^{n,S}-\ell(q_{n^{\star}}^{n,q,v})\right)\right)\right],
\end{align}
where $\mathcal{A}_{n}$ is the set of admissible strategies at time
$t_n$ defined as:
$$
\begin{array}{cccl}
\mathcal{A}_{n} & = & \Big\{(v,n^{\star})\Big| & v=(v_{k})_{n\le k\le n^{\star}-1}\text{ is }\left(\mathcal{F}\right)\text{-adapted},\\
 &  &  & n^{\star}\text{ is a }\left(\mathcal{F}\right)\text{-stopping time taking values in }\left(\mathcal{N}\cup\{N\}\right)\cap\{n,\ldots,N\}\Big\},
\end{array}
$$
and where the state variables are defined for $0\leq n\leq k\leq N$
and $k>0$ by:
$$
\left\{ \begin{array}{lcl}
X_{k}^{n,x,v} & = & x+\Sum_{j=n}^{k-1}v_{j}S_{j+1}^{n,S}\delta t+L\left(\frac{v_{j}}{V_{j+1}}\right)V_{j+1}\delta t\\
q_{k}^{n,q,v} & = & q-\Sum_{j=n}^{k-1}v_{j}\delta t\\
S_{k}^{n,S} & = & S+\sigma\sqrt{\delta t}\Sum_{j=n}^{k-1}\epsilon_{j+1}\\
A_{k}^{n,A,S} & = & \frac{n}{k}A+\frac{1}{k}\Sum_{j=n}^{k-1}S_{j+1}^{n,S}
\end{array}\right.
$$

\begin{rem}
We can restrict the value function to $q\in[0,Q]$. It is indeed straightforward
to see that strategies involving values of $q$ outside of $[0,Q]$
are suboptimal.
\end{rem}
We now introduce the Bellman equation associated to the dynamic problem \eqref{def maximisation dynamique}. For that purpose, we introduce for $ n>N$:
\begin{align}
\tilde{u}_{n,n+1}(x,q,S,A) =  \sup_{v\in\R}\E\left[u_{n+1}\Big(X_{n+1}^{n,x,v},q_{n+1}^{n,q,v},S_{n+1}^{n,S},A_{n+1}^{n,A,S}\Big)\right].\label{eq:def-utilde}
\end{align}
We then have:

\begin{prop}
Consider the family of functions $ (u_n)_{0 \le n \le N}$ defined by \eqref{def maximisation dynamique}. Then it is the unique solution of its associated Bellman equation:
\begin{eqnarray}
u_{n}(x,q,S,A) & = & \begin{cases}
-\exp\left(-\gamma\left(QA-x-qS-\ell(q)\right)\right) & \text{if }n=N,\\
\begin{array}{rl}
\max\bigg\{ & \tilde{u}_{n,n+1}\left(x,q,S,A\right),\\
 & -\exp\left(-\gamma\left(QA-x-qS-\ell(q)\right)\right)\bigg\}
\end{array} & \text{if }n\in\mathcal{N},\\
\tilde{u}_{n,n+1}\left(x,q,S,A\right) & \text{otherwise}.
\end{cases}\label{eq:ppd-u}
\end{eqnarray}
\end{prop}

We refer to \cite{bell} for a proof of this result.
\begin{rem}
The initial time plays a specific role here, as for $n=0$,
$$
A_{n+1}^{n,A,S}=A_{1}^{0,A,S}=S_{1}^{0,S}=S+\sigma\sqrt{\delta t}\epsilon_{1}.
$$
Subsequently, $u_{0}$ does not depend on $A$ and writes:
$$
u_{0}(x,q,S)  = \sup_{v\in\R}\E\left[ u_{1}\Big(X_{1}^{0,x,v},q_{1}^{0,q,v},S_{1}^{0,S},S_{1}^{0,S}\Big)\right].
$$
\end{rem}

\section{Reduction to a 3-variable equation}
\label{sect:reduc-var}
\subsection{Change of variables}
\label{sect:change of var}
In order to reduce the dimensionality of the problem, we introduce
in Proposition \ref{prop:cdv} a change of variables that consists
mainly in writing the problem in terms of the (normalized) spread $Z:=\frac{S-A}{\sigma \sqrt{\delta t}}$.

Before that, let us start with a Lemma.
\begin{lem}
\label{lem:y}For $1\le n\le k\le N$, if $Y_{k}:=Y_{k}^{n,X,A,S,v}$
and $Y$ are defined as
$$
\begin{array}{ccclclcl}
Y_{k} & = & - & QA_{k}^{n,A,S} & + & X_{k}^{n,x,v} & + & q_{k}^{n,q,v}S_{k}^{n,S},\\
Y & = & - & QA & + & x & + & qS,
\end{array}
$$
then we have:
\begin{eqnarray*}
Y_{k}-Y & = & \sigma\sqrt{\delta t}\left(\sum_{j=n}^{k-1}\left(q_{j}-\left(1-\frac{j}{k}\right)Q\right)\epsilon_{j+1}-\left(1-\frac{n}{k}\right)Q\dfrac{S-A}{\sigma\sqrt{\delta t}}\right)+\sum_{j=n}^{k-1}L\left(\frac{v_{j}}{V_{j+1}}\right)V_{j+1}\delta t.
\end{eqnarray*}
\end{lem}
\begin{proof}
We omit all the superscripts for the sake of readability. On the one hand, we have:
\begin{align*}
A_{k}-A &=  \left(\frac{n}{k}A+\frac{1}{k}\sum_{j=n}^{k-1}S_{j+1}\right)-A = \left(1-\frac{n}{k}\right)\left(S-A\right)+\sigma\sqrt{\delta t}\sum_{j=n}^{k-1}\left(1-\frac{j}{k}\right)\epsilon_{j+1}.
\end{align*}
On the other hand, we have:
\begin{align*}
X_{k}-x+q_{k}S_{k}-qS &= \sum_{j=n}^{k-1}\left(q_{j}-q_{j+1}\right)S_{j+1}+\sum_{j=n}^{k-1}L\left(\frac{v_{j}}{V_{j+1}}\right)V_{j+1}\delta t+q_{k}S_{k}-qS\\
&=\sum_{j=n}^{k-1}q_{j}\left(S_{j+1}-S_{j}\right)+\sum_{j=n}^{k-1}L\left(\frac{v_{j}}{V_{j+1}}\right)V_{j+1}\delta t\\
 &=\sum_{j=n}^{k-1}q_{j}\sigma\sqrt{\delta t}\epsilon_{j+1}+\sum_{j=n}^{k-1}L\left(\frac{v_{j}}{V_{j+1}}\right)V_{j+1}\delta t.
\end{align*}
Combining  these equations we obtain the result.
\end{proof}

We now introduce the key change of variables for the next Proposition. For $ n\geq 1$, $q\in[0,Q]$, and $ Z \in \mathbb{R}$, we define:
\begin{align}
\theta_{n}(q,Z) = & \inf_{(v,n^{\star})\in\mathcal{A}_{n}}\logexpo\sigma\sqrt{\delta t}\left(\sum_{j=n}^{n^{\star}-1}\left(q_{j}-\left(1-\frac{j}{n^{\star}}\right)Q\right)\epsilon_{j+1}\right. \nonumber \\
& \left. -\left(1-\frac{n}{n^{\star}}\right)QZ\right) +\sum_{j=n}^{n^{\star}-1}L\left(\frac{v_{j}}{V_{j+1}}\right)V_{j+1}\delta t+\ell(q_{n^{\star}}^{n,q,v})\logexpc ,\label{eq:def-theta}
\end{align}
and
\begin{align}
\tilde{\theta}_{n,n+1}\left(q,Z\right) =  & \inf_{v\in\R}\logexpo\sigma\sqrt{\delta t}\left(\left(q-\dfrac{Q}{n+1}\right)\epsilon_{n+1}-\dfrac{Q}{n+1}Z\right)\nonumber \\
 & +L\left(\frac{v}{V_{n+1}}\right)V_{n+1}\delta t+\theta_{n+1}\left(q-v\delta t,\dfrac{n}{n+1}\left(Z+\epsilon_{n+1}\right)\right)\logexpc.\label{eq:def-thetatilde}
\end{align}
The result is then the following:
\begin{prop}
\label{prop:cdv}For $n\geq 1$, we have the following characterization of problem (\ref{def maximisation dynamique}):
\\
{\rm (i)} $u_{n}(x,q,S,A)$ verifies
\begin{eqnarray}
u_{n}(x,q,S,A) & = & -\exp\left(-\gamma\left(-Y-\theta_{n}\left(q,\dfrac{S-A}{\sigma\sqrt{\delta t}}\right)\right)\right).\label{eq:cdv}
\end{eqnarray}
\\
{\rm (ii)}
The Bellman equation satisfied by $\theta_{n}$ is:
\begin{eqnarray}
\theta_{n}(q,Z) & = & \begin{cases}
\ell(q) & \text{if }n=N,\\
\min\bigg\{\tilde{\theta}_{n,n+1}(q,Z),\ell(q)\bigg\} & \text{if }n\in\mathcal{N},\\
\tilde{\theta}_{n,n+1}(q,Z) & \text{otherwise}.
\end{cases}\label{eq:bellman-theta}
\end{eqnarray}
\end{prop}

\begin{proof}
For (i), let us first notice that:
\begin{align*}
  \frac{1}{\gamma}\log\big(-u_{n}(x,q,S,A)\big)-Y & =\inf_{(v,n^{\star})\in\mathcal{A}}\logexpo Y_{n^{\star}}-Y+\ell(q_{n^{\star}}^{n,q,v})\logexpc.
\end{align*}
Using Lemma \ref{lem:y}, we compute the difference $Y_{n^{\star}}-Y$:
\begin{align*}
Y_{n^{\star}}-Y = \sigma\sqrt{\delta t}\left(\sum_{j=n}^{n^{\star}-1}\left(q_{j}-\frac{n^{\star}-j}{n^{\star}}Q\right)\epsilon_{j+1}-\frac{n^\star-n}{n^{\star}}Q\dfrac{S-A}{\sigma\sqrt{\delta t}}\right) +\sum_{j=n}^{n^{\star}-1}L\left(\frac{v_{j}}{V_{j+1}}\right)V_{j+1}\delta t.
\end{align*}
Then
\begin{eqnarray*}
 &  & \frac{1}{\gamma}\log\left(-u_{n}(x,q,S,A)\right)-Y\\
 & = & \inf_{(v,n^{\star})\in\mathcal{A}}\logexpo\sigma\sqrt{\delta t}\left(\sum_{j=n}^{n^{\star}-1}\left(q_{j}-\left(1-\frac{j}{n^{\star}}\right)Q\right)\epsilon_{j+1}-\left(1-\frac{n}{n^{\star}}\right)Q\dfrac{S-A}{\sigma\sqrt{\delta t}}\right)\\
 &  & +\sum_{j=n}^{n^{\star}-1}L\left(\frac{v_{j}}{V_{j+1}}\right)V_{j+1}\delta t+\ell(q_{n^{\star}}^{n,q,v})\logexpc.
\end{eqnarray*}
Using the definition of $\theta_{n}$, we obtain $u_{n}(x,q,S,A) = -\exp\left(-\gamma\left(-Y-\theta_{n}\left(q,\dfrac{S-A}{\sigma\sqrt{\delta t}}\right)\right)\right).$\\

To prove (ii), we first show that $\tilde{\theta}_{n,n+1}$ defined in (\ref{eq:def-thetatilde})
satisfies the following equation:
\begin{align}
\tilde{u}_{n,n+1}(x,q,S,A) = -\exp\left(-\gamma\left(-Y-\tilde{\theta}_{n,n+1}\left(q,\dfrac{S-A}{\sigma\sqrt{\delta t}}\right)\right)\right).\label{eq:cdvtilde}
\end{align}
For that purpose, we notice that:
\begin{eqnarray*}
I & =  & \frac{1}{\gamma}\log\left(-\tilde{u}_{n,n+1}(x,q,S,A)\right)-Y\\
 & = & \inf_{v\in\R}\logexpo Y_{n+1}-Y+\theta_{n+1}\left(q-v\delta t,\dfrac{S_{n+1}-A_{n+1}}{\sigma\sqrt{\delta t}}\right)\logexpc.
\end{eqnarray*}
The difference $Y_{n+1}-Y$ can be computed using Lemma \ref{lem:y}
as above:
$$
Y_{n+1}-Y = \sigma\sqrt{\delta t}\left(\left(q-\dfrac{Q}{n+1}\right)\epsilon_{n+1}-\dfrac{Q}{n+1}\dfrac{S-A}{\sigma\sqrt{\delta t}}\right)+L\left(\frac{v_{n}}{V_{n+1}}\right)V_{n+1}\delta t.
$$
The difference $S_{n+1}-A_{n+1}$ can be computed directly:
$$
 S_{n+1}-A_{n+1} = S_{n+1}-\left(\frac{n}{n+1}A+\frac{1}{n+1}S_{n+1}\right)
 = \frac{n}{n+1}\sigma\sqrt{\delta t}\left(\dfrac{S-A}{\sigma\sqrt{\delta t}}+\epsilon_{n+1}\right).
$$
Therefore, we have
\begin{eqnarray*}
I & = & \inf_{v\in\R}\logexpo\sigma\sqrt{\delta t}\left(\left(q-\dfrac{Q}{n+1}\right)\epsilon_{n+1}-\dfrac{Q}{n+1}\dfrac{S-A}{\sigma\sqrt{\delta t}}\right)\\
 &  & +L\left(\frac{v_{n}}{V_{n+1}}\right)V_{n+1}\delta t+\theta_{n+1}\left(q-v\delta t,\frac{n}{n+1}\left(\dfrac{S-A}{\sigma\sqrt{\delta t}}+\epsilon_{n+1}\right)\right)\logexpc.
\end{eqnarray*}
This gives (\ref{eq:cdvtilde}). Finally, we plug the equations (\ref{eq:cdv}) and (\ref{eq:cdvtilde})
into (\ref{eq:ppd-u}) to obtain (\ref{eq:bellman-theta}).\\
The proof that $ (\theta_n)$ is bounded, allowing us to compute directly this dynamic programming equation from \eqref{eq:ppd-u}, is in Appendix A.
\end{proof}

\subsection{Rationale behind the change of variable and comparison with \cite{jaimungal}}
Proposition \ref{prop:cdv} states that the problem to be solved is in fact
of dimension 3, the three dimensions being the time, the inventory
$q$, and the spread $Z=\frac{S-A}{\sigma\sqrt{\delta t}}$. In particular, the buying behavior of the bank depends on $S$ and $A$ only through $Z$.\\

It is interesting to notice the link between this result and the one obtained by Jaimungal \emph{et al.} in \cite{jaimungal}. In a different framework, the optimal strategy obtained in \cite{jaimungal} only depends indeed on the ratio $\frac{S}{A}$. These results are important as they permit to reduce the number of variables, leading to a significant increase of the speed of numerical methods. In both cases, the relevant variable is the relative position of the price $S$ with respect to the average price $A$. This is in line with the intuition: if $S$ goes below $A$, then there is an incentive to buy shares as the gain in the absence of liquidity costs would be $A-S$. Finally, we highlight that their change of variables is a ratio and not a difference as in our framework. This comes from their choice of lognormal price returns. They consider indeed a continuous time framework with a geometric Brownian motion for $S$, while we chose a discrete time framework with independent and normally distributed price increments.\\

To find the optimal liquidation strategy and the optimal stopping time, we
need to solve a recursive equation to compute $(q,Z)\mapsto(\theta_{n}(q,Z))_{n}$.
The definition of $\theta_{n}$ deserves a few comments. It is made
of three parts:
\begin{align}
\label{effects}
\theta_{n}(q,Z) = &\inf_{(v,n^{\star})\in\mathcal{A}_{n}}\logexpo\sigma\sqrt{\delta t}\Biggl(\underbrace{\sum_{j=n}^{n^{\star}-1}\left(q_{j}-\left(1-\frac{j}{n^{\star}}\right)Q\right)\epsilon_{j+1}}_{\text{risk\ term}}-\underbrace{\left(1-\frac{n}{n^{\star}}\right)QZ}_{Z\text{\ term}}\Biggr)\nonumber \\
 & +\underbrace{\sum_{j=n}^{n^{\star}-1}L\left(\frac{v_{j}}{V_{j+1}}\right)V_{j+1}\delta t+\ell(q_{n^{\star}}^{n,q,v})}_{\text{liquidity\ term}}\logexpc.
\end{align}
The risk term corresponds to the risk associated to the payoff. If
$n^{\star}$ was fixed, this risk could be hedged perfectly, by buying
the shares evenly until time $t_{n^{\star}} = n^\star \delta t$. However, $n^{\star}$
is a stopping time and it is not known ex-ante. Practically, this
means that, to hedge the risk associated to the payoff, the strategy
depends (roughly) on the targeted value of $n^{\star}$, a targeted
value that changes according to the evolution of $Z$: when $S$ is
below $A$, we expect to end the process early to benefit from the
difference between $A$ and $S$, whereas we expect $n^{\star}$ to
be large if $S$ is greater then $A$.\\
To understand the $Z$ term, it is worth recalling that the certainty equivalent of the bank at time $t_n$ is $QA_n -X_n -q_nS_n - \theta_{n}\left(q_n,\dfrac{S_n-A_n}{\sigma\sqrt{\delta t}}\right)$ (this is equation (\ref{eq:cdv})). In other words, since the bank would earn $QA_n -X_n -q_nS_n$ if it were able to buy $q_n$ shares now and to deliver (and exercise the option) instantaneously, the function $\theta_n$ represents the cost linked to the limited available liquidity. For negative values of $Z$, the $Z$ term quantifies the fact that because of the limited available liquidity, and because $Z$ is mean-reverting around $0$, the bank will not be able to benefit entirely from the current difference between $S$ and $A$. $\theta_n$ is then a decreasing function of $Z$ and the further the targeted time of delivery $n^\star$, the more $\theta_n$ is decreasing in $Z$, as $Z$ will have time to return to $0$. Similarly, the fact that $\theta_n$ is decreasing in $Z$ when $Z$ is positive is straightforward as the situation will naturally improve for positive values of $Z$, because of the mean-reversion phenomenon.\\
The last two terms correspond to (running) execution costs and liquidity costs at final time.

\subsection{Optimal strategy and pricing of ASR contracts}
\label{sect: Optimal strategy}

In this section we insist on the price of the ASR contract and how to optimally hedge it with the use of Proposition \ref{prop:cdv}.

Let us come back to our initial problem. The value
function of the bank at inception is
\begin{align*}
 u_0(0,Q,S_0) &=\sup_{v\in\R } \E \Big[ u_1 \Big( X_{1}^{0,0,v},q_{1}^{0,Q,v},S_{1}^{0,S_{0}},S_{1}^{0,S_{0}} \Big)\Big]\\
 & = \sup_{v\in\R} \E\Big[-\exp\left(-\gamma\left(QS_{1}^{0,S_{0}}-X_{1}^{0,0,v}-QS_{1}^{0,S_{0}}-\theta_{1}\left(Q-v\delta t,0\right)\right)\right) \Big]\\
  & = \sup_{v\in\R}\E\left[-\exp\left(-\gamma\left(-L\left(\dfrac{v}{V_{1}}\right)V_{1}\delta t-\theta_{1}\left(Q-v\delta t,0\right)\right)\right)\right].
\end{align*}
Hence, the value function of the bank at time $0$ does not depend
on the price.

This allows us to obtain the following formulation for the indifference price of the ASR contract from $ \theta_1$:
\begin{prop}
The indifference price of an ASR contract is given by:
$$ \Pi=\frac{1}{\gamma} \log\left(-u_0(0,Q,S_0) \right)=\inf_{v\in\R}\left\{ L\left(\dfrac{v}{V_{1}}\right)V_{1}\delta t+\theta_{1}\left(Q-v\delta t,0\right)\right\}.$$
\end{prop}

Pricing an ASR contract therefore relies on solving the recursive equations (\ref{eq:bellman-theta}) for $(\theta_{n})_n$.

\begin{rem}
We can generalize the above definition of $\theta_{n}$ to the case
$n=0$. We then see straightforwardly that $\Pi=\theta_{0}(Q,0)$.
\end{rem}

Associated to this indifference price, the functions $(\theta_{n})_{n}$
and the recursive equations (\ref{eq:bellman-theta}) permit to find
an optimal strategy. We describe the optimal strategy in the Proposition below:
\begin{prop}
At time $0$, an optimal strategy consists in sending an order of
size $v^{*}\delta t$ where
$$
v^{*}\in\argmin_{v\in\R}L\left(\dfrac{v}{V_{1}}\right)V_{1}\delta t+\theta_{1}\left(Q-v\delta t,0\right).
$$
At an intermediate time $t_n<T$, if the bank has not already
delivered the shares, then, denoting $Z_{n}=\frac{S_{n}-A_{n}}{\sigma\sqrt{\delta t}}$,
an optimal strategy consists in the following:
\begin{itemize}
\item If $n\not\in\mathcal{N}$, send an order of size $v^{*}\delta t$
where
\begin{eqnarray*}
v^{*} & \in & \argmin_{v\in\R}\logexpo\sigma\sqrt{\delta t}\left(\left(q_{n}-\dfrac{Q}{n+1}\right)\epsilon_{n+1}-\dfrac{Q}{n+1}Z_{n}\right)\\
 &  & +L\left(\frac{v}{V_{n+1}}\right)V_{n+1}\delta t+\theta_{n+1}\left(q-v\delta t,\dfrac{n}{n+1}\left(Z_{n}+\epsilon_{n+1}\right)\right)\logexpc.
\end{eqnarray*}
\item If $n\in\mathcal{N}$, compare $\tilde{\theta}_{n,n+1}(q_{n},Z_{n})$
and $\ell(q_{n})$.
\begin{itemize}
\item If $\tilde{\theta}_{n,n+1}(q_{n},Z_{n})<\ell(q_{n})$, then send an
order of size $v^{*}\delta t$ where
\begin{eqnarray*}
v^{*} & \in & \argmin_{v\in\R}\logexpo\sigma\sqrt{\delta t}\left(\left(q_{n}-\dfrac{Q}{n+1}\right)\epsilon_{n+1}-\dfrac{Q}{n+1}Z_{n}\right)\\
 &  & +L\left(\frac{v}{V_{n+1}}\right)V_{n+1}\delta t+\theta_{n+1}\left(q-v\delta t,\dfrac{n}{n+1}\left(Z_{n}+\epsilon_{n+1}\right)\right)\logexpc.
\end{eqnarray*}
\item If $\ell(q_{n})\le\tilde{\theta}_{n,n+1}(q_{n},Z_{n})$, then deliver
the shares after the remaining ones have been bought (supposedly instantaneously
in the model, at price $S_nq_n + \ell(q_n)$).
\end{itemize}
\end{itemize}
\end{prop}

\section{Introducing permanent market impact}
\label{sect: permanent market impact}
\subsection{Setup of the model}

In the initial setup presented above, market impact was only temporary
and boiled down to execution costs. In this section we generalize
the previous model to introduce permanent market impact.

In order to introduce permanent market impact, we need to approximate
the payoff of the ASR: instead of considering that $S_{n+1}$ is the
VWAP over the day $[t_n,t_{n+1}]$, we assume that it
is the closing price of that day and hence that $A$ is the average
of the closing prices since inception. This approximation is acceptable
and often made by practitioners when they consider ASR contracts.

The number of shares remaining to be bought still evolves as
$$ q_{n+1}=q_{n}-v_{n}\delta t, $$
but its evolution impacts the stock price.

In most papers, following \cite{almgren} and \cite{gatheral}, the
impact on price is assumed to be proportional to $v_{n}\delta t$.
Here, in line with the square-root law for market impact, we consider
a more general form that is the discrete counterpart of the dynamic-arbitrage-free
model proposed in \cite{gueantperm} (see also the work done by Alfonsi \emph{et al.} in \cite{alfonsi}). For that purpose, we introduce
a nonnegative and decreasing function%
\footnote{This function is a constant function in the case of a linear permanent market
impact. In that case, the impact only depends on the trading quantity $q_{n+1}-q_n = v_n \delta t$
} $f\in L_{loc}^{1}(\R_{+})$, and we write
$$ S_{n+1}=S_{n}+\sigma\sqrt{\delta t}\epsilon_{n+1}+\left(G(q_{n+1})-G(q_{n})\right), $$
where $G(q)=\int_{q}^{Q}f(|Q-y|)dy$.

If we assume that the order of size $v_{n}\delta t$ is executed evenly over the interval $[t_n,t_{n+1}]$, this
leads to the following dynamics for the cash account (see appendix B):
$$
X_{n+1}=X_{n}+S_{n+1}v_{n}\delta t-q_{n}\left(G(q_{n+1})-G(q_{n})\right)+\left(F(q_{n+1})-F(q_{n})\right)-\frac{\sigma v_{n}\delta t^{\frac{3}{2}}}{\sqrt{3}}\epsilon'_{n+1},
$$
where $\left((\epsilon_{k},\epsilon'_{k}\right))_{k}$ are i.i.d.
random variables with moment generating function defined on $\R_+$, with $\mathbb{E}\left[(\epsilon_{k},\epsilon'_{k}\right)]=0$ and
$\mathbb{V}[(\epsilon_{k},\epsilon'_{k})]=\left(\begin{array}{cc}
1 & \dfrac{\sqrt{3}}{2}\\
\dfrac{\sqrt{3}}{2} & 1
\end{array}\right)$ and where $F(q)=\int_{q}^{Q}yf(|Q-y|)dy$.

\begin{rem}
The terms involving $F$ and $G$ are linked to permanent market impact
and stand for the fact that market impact is felt progressively over
the course of the execution of the order. The noise term $\frac{\sigma v_{n}\delta t^{\frac{3}{2}}}{\sqrt{3}}\epsilon'_{n+1}$
corresponds to the fact that we execute progressively and evenly over the day while the price
$S_{n+1}$ is the closing price. We refer the reader to Appendix B for details about this model.
\end{rem}
So far, we considered permanent market impact over the course of the
buying process but not at time of delivery. To account for it at time
of delivery, we suppose that, at time $n^{\star}\in\mathcal{N}\cup\{N\}$,
the shares remaining to be bought ($q_{n^{\star}}$) can be purchased
for a total spending of
\footnote{See appendix B for the rationale of this definition that guarantees the absence of dynamic arbitrage.} $q_{n^{\star}}S_{n^{\star}}+F(0)-F(q_{n^{\star}})+\ell(q_{n^{\star}})$
where $\ell$ is, as above, a penalty function, assumed to be convex,
even, increasing on $\R_{+}$ and verifying $\ell(0)=0$.

The optimization problem the bank faces is then slightly different:
\begin{equation}
\sup_{(v,n^{\star})\in\mathcal{A}}\E\left[-\exp\left(-\gamma\left(QA_{n^{\star}}-X_{n^{\star}}-q_{n^{\star}}S_{n^{\star}}-F(0)+F(q_{n^{\star}})-\ell(q_{n^{\star}})\right)\right)\right],\label{defmaximisation2}
\end{equation}

\subsection{A similar reduction to a 3-dimension problem}

To solve the problem, we introduce as above the value function of
the problem at time $t_n$:
\begin{eqnarray*}
u_{n}(x,q,S,A) & = & \sup_{(v,n^{\star})\in\mathcal{A}_{n}}\E\Big[-\exp\Big(-\gamma\Big(QA_{n^{\star}}^{n,A,S}-X_{n^{\star}}^{n,x,v}-q_{n^{\star}}^{n,q,v}S_{n^{\star}}^{n,S,v}\\
 &  & -F(0)+F(q_{n^{\star}}^{n,q,v})-\ell(q_{n^{\star}}^{n,q,v})\Big)\Big)\Big],
\end{eqnarray*}
where the state variables are defined for $0\leq n\leq k\leq N$ and
$k>0$ by:
\[
\left\{ \begin{array}{lcl}
X_{k}^{n,x,v} & = & x+(F(q_{k}^{n,q,v})-F(q))+\Sum_{j=n}^{k-1}v_{j}S_{j+1}^{n,S,v}\delta t\\
 &  & -q_{j}^{n,q,v}\left(G(q_{j+1}^{n,q,v})-G(q_{j}^{n,q,v})\right)+L\left(\frac{v_{j}}{V_{j+1}}\right)V_{j+1}\delta t\\
q_{k}^{n,q,v} & = & q-\Sum_{j=n}^{k-1}v_{j}\delta t\\
S_{k}^{n,S,v} & = & S+G(q_{k}^{n,q,v})-G(q)+\sigma\sqrt{\delta t}\Sum_{j=n}^{k-1}\epsilon_{j+1}\\
A_{k}^{n,A,S} & = & \frac{n}{k}A+\frac{1}{k}\Sum_{j=n}^{k-1}S_{j+1}^{n,S,v}
\end{array}\right.
\]

The Bellman equation associated to this problem is the following:
\begin{eqnarray}
u_{n}(x,q,S,A) & = & \begin{cases}
-\exp\left(-\gamma\left(QA-x-qS-F(0)+F(q)-\ell(q)\right)\right) & \text{if }n=N,\\
\max\bigg\{\tilde{u}_{n,n+1}\left(x,q,S,A\right),\\
-\exp\left(-\gamma\left(QA-x-qS-F(0)+F(q)-\ell(q)\right)\right)\bigg\}, & \text{if }n\in\mathcal{N},\\
\tilde{u}_{n,n+1}\left(x,q,S,A\right) & \text{otherwise},
\end{cases}\label{eq:ppd-u2}
\end{eqnarray}
where
\begin{eqnarray}
\tilde{u}_{n,n+1}(x,q,S,A) & = & \sup_{v\in\R}\E\Big[u_{n+1}\Big(X_{n+1}^{n,x,v},q_{n+1}^{n,q,v},S_{n+1}^{n,S,v},A_{n+1}^{n,A,S}\Big)\Big].\label{eq:def-utilde2}
\end{eqnarray}
Using a change of variables similar to the one used in Section \ref{sect:reduc-var},
we obtain the reduction of the problem to 3 dimensions:%
\footnote{The proof proceeds from the same ideas as for Proposition \ref{prop:cdv}.%
}
\begin{prop}
\label{prop:cdv2}For $n\ge1$, $u_{n}(x,q,S,A)$ can be written as
\begin{eqnarray}
u_{n}(x,q,S,A) & = & -\exp\left(-\gamma\left(QA-x-qS+F(q)-\theta_{n}\left(q,\dfrac{S-A}{\sigma\sqrt{\delta t}}\right)\right)\right),\label{eq:cdv2}
\end{eqnarray}
where $\theta_{n}$ is defined by induction as:
\begin{eqnarray}
\theta_{n}(q,Z) & = & \begin{cases}
\ell(q)+F(0) & \text{if }n=N,\\
\min\bigg\{\tilde{\theta}_{n,n+1}(q,Z),\ell(q)+F(0)\bigg\} & \text{if }n\in\mathcal{N},\\
\tilde{\theta}_{n,n+1}(q,Z) & \text{otherwise},
\end{cases}\label{eq:bellman-theta2}
\end{eqnarray}
with:
\begin{eqnarray}
\tilde{\theta}_{n,n+1}\left(q,Z\right) & = & \inf_{v\in\R}\logexpo\sigma\sqrt{\delta t}\left(\left(q-\dfrac{Q}{n+1}\right)\epsilon_{n+1}-\dfrac{Q}{n+1}Z\right)-\frac{\sigma v\delta t^{\frac{3}{2}}}{\sqrt{3}}\epsilon'_{n+1}\nonumber \\
 &  & -\frac{Q}{n+1}\left(G(q-v\delta t)-G(q)\right)+L\left(\frac{v}{V_{n+1}}\right)V_{n+1}\delta t\nonumber \\
 &  & +\theta_{n+1}\left(q-v\delta t,\dfrac{n}{n+1}\left(Z+\epsilon_{n+1}+\frac{G(q-v\delta t)-G(q)}{\sigma\sqrt{\delta t}}\right)\right)\logexpc.\label{eq:def-thetatilde2}
\end{eqnarray}

\end{prop}

\subsection{Optimal strategy and pricing of ASR contracts}

As in section \ref{sect: Optimal strategy}, we can deduce from $(\theta_{n})_{n}$ and the
recursive equations (\ref{eq:bellman-theta2}), both the price of
an ASR contract and an optimal strategy.

The indifference price $\Pi$ is implicitly defined by:
\begin{eqnarray*}
u_{0}(0,Q,S_{0}) & = & -\exp\left(-\gamma\left(-\Pi\right)\right).
\end{eqnarray*}

Following the previous calculations, we obtain:

\begin{prop}
The indifference price $\Pi$ of the ASR contract in presence of permanent market impact is given by:
\begin{eqnarray*}
\Pi & = & \inf_{v\in\R}\logexpo-\frac{ \sigma v\delta t^{\frac{3}{2}}}{\sqrt{3}}\epsilon'_{1}-Q\left(G(Q-v\delta t)-G(Q)\right)\\
    &   &  + L\left(\frac{v}{V_{1}}\right)V_{1}\delta t+\theta_{1}\left(Q-v\delta t,0\right)\logexpc\\
 & = & \inf_{v\in\R} \Big\{ \frac{1}{\gamma}h\left(-\gamma\frac{v\delta t}{\sqrt{3}}\right)-Q\left(G(Q-v\delta t)-G(Q)\right)+L\left(\frac{v}{V_{1}}\right)V_{1}\delta t+\theta_{1}\left(Q-v\delta t,0\right) \Big\} ,
\end{eqnarray*}
where $h$ is the cumulant-generating function of the random variable
$\sigma \sqrt{\delta t} \epsilon_{1}'$.
\end{prop}

Associated to this indifference price, we can exhibit an optimal strategy.
\begin{prop}
At time $0$, an optimal strategy consists in sending an order of
size $v^{*}\delta t$ where
\begin{eqnarray*}
v^{*} & \in & \argmin_{v\in\R} \Big\{ \frac{1}{\gamma}h\left(-\gamma\frac{v\delta t}{\sqrt{3}}\right)-Q\left(G(Q-v\delta t)-G(Q)\right)+L\left(\frac{v}{V_{1}}\right)V_{1}\delta t+\theta_{1}\left(Q-v\delta t,0\right) \Big\}.
\end{eqnarray*}
At an intermediate time $t_n<T$, if the bank has not already
delivered the shares, then, denoting $Z_{n}=\frac{S_{n}-A_{n}}{\sigma\sqrt{\delta t}}$,
an optimal strategy consists in the following:
\begin{itemize}
\item If $n\not\in\mathcal{N}$, send an order of size $v^{*}\delta t$
where
\begin{eqnarray*}
v^{*} & \in & \argmin_{v\in\R}\logexpo\sigma\sqrt{\delta t}\left(\left(q-\dfrac{Q}{n+1}\right)\epsilon_{n+1}-\dfrac{Q}{n+1}Z\right)-\frac{\sigma v\delta t^{\frac{3}{2}}}{\sqrt{3}}\epsilon'_{n+1}\\
 &  & -\frac{Q}{n+1}\left(G(q-v\delta t)-G(q)\right)+L\left(\frac{v}{V_{n+1}}\right)V_{n+1}\delta t\\
 &  & +\theta_{n+1}\left(q-v\delta t,\dfrac{n}{n+1}\left(Z+\epsilon_{n+1}+\frac{G(q-v\delta t)-G(q)}{\sigma\sqrt{\delta t}}\right)\right)\logexpc.
\end{eqnarray*}

\item If $n\in\mathcal{N}$, compare $\tilde{\theta}_{n,n+1}(q_{n},Z_{n})$
and $\ell(q_{n})+F(0)$.

\begin{itemize}
\item If $\tilde{\theta}_{n,n+1}(q_{n},Z_{n})<\ell(q_{n})+F(0)$, then send
an order of size $v^{*}\delta t$ where
\begin{eqnarray*}
v^{*} & \in & \argmin_{v\in\R}\logexpo\sigma\sqrt{\delta t}\left(\left(q-\dfrac{Q}{n+1}\right)\epsilon_{n+1}-\dfrac{Q}{n+1}Z\right)-\frac{\sigma v\delta t^{\frac{3}{2}}}{\sqrt{3}}\epsilon'_{n+1}\\
 &  & -\frac{Q}{n+1}\left(G(q-v\delta t)-G(q)\right)+L\left(\frac{v}{V_{n+1}}\right)V_{n+1}\delta t\\
 &  & +\theta_{n+1}\left(q-v\delta t,\dfrac{n}{n+1}\left(Z+\epsilon_{n+1}+\frac{G(q-v\delta t)-G(q)}{\sigma\sqrt{\delta t}}\right)\right)\logexpc.
\end{eqnarray*}

\item If $\ell(q_{n})+F(0)\le\tilde{\theta}_{n,n+1}(q_{n},Z_{n})$, then
deliver the shares after the remaining ones have been bought (supposedly
instantaneously in the model).
\end{itemize}
\end{itemize}
\end{prop}

\section{Numerical methods and examples}
\label{sect: numerics}
We now present numerical methods in the case where there is no permanent market impact. In this case indeed, the problem can be solved using trees of polynomial size if we specify the law of $(\epsilon_n)_n$ adequately. These trees are not classical recombinant trees as in most tree methods used in Finance (e.g. Cox-Ross-Rubinstein model) but the total number of nodes evolves as $\mathcal{O}(N^3)$ where $N$ is the number of time steps. The model we present is based on pentanomial trees. Using these trees, we  carry out comparative statics in order to analyze the role of the parameters.

\subsection{Description of the pentanomial tree method}

To design a numerical method in order to obtain a solution to our problem, we need first to choose a specific distribution for the innovation process $\left(\epsilon_{n}\right)_{n\ge1}$. A simple case would consist in a simple coin-toss model with $\epsilon_{n}=\pm 1$ with probabilities 50\%-50\%. In order to obtain more complex price dynamics, we consider the following distribution for  $\left(\epsilon_{n}\right)_{n\ge1}$:\footnote{The distribution is chosen to match the first four moments of a standard normal: \[
\left\{ \begin{array}{ccc}
\mathbb{E}\left[\epsilon_{n}\right] & = & 0\\
\mathbb{E}\left[\epsilon_{n}^{2}\right] & = & 1\\
\mathbb{E}\left[\epsilon_{n}^{3}\right] & = & 0\\
\mathbb{E}\left[\epsilon_{n}^{4}\right] & = & 3
\end{array}\right.
\]}

\begin{eqnarray*}
\epsilon_{n} & = & \begin{cases}
+2 & \text{with probability }\frac{1}{12}\\
+1 & \text{with probability }\frac{1}{6}\\
0 & \text{with probability }\frac{1}{2}\\
-1 & \text{with probability }\frac{1}{6}\\
-2 & \text{with probability }\frac{1}{12}
\end{cases}
\end{eqnarray*}

Using the Bellman equations  (\ref{eq:bellman-theta}) and (\ref{eq:def-thetatilde}), we see that the computation of $\theta_{n}(q,Z)$ for $n<N$ requires the values of $\theta_{n+1}\left(\cdot,\dfrac{n}{n+1}\left(Z+k\right)\right)$
for $k$ in $\{-2,-1,0,+1,+2\}$.

In order to solve numerically the Bellman equation, we therefore use a tree. The issue is that nothing ensures a priori that the number of nodes does not increase exponentially with the number of time steps. In fact, we will show that the growth in the number of nodes is quadratic\footnote{This means that the number of nodes in the whole tree is a cubic function of the number of time steps $N$.} rather than exponential in the number of time steps. To see this point, we start with a change of variables.

\begin{prop}
$\theta_{n}$ can be written as:
\begin{eqnarray*}
\theta_{n}\left(q,Z\right) & = & \Theta_{n}\left(q,n\left(Z+n-1\right)\right),
\end{eqnarray*}
where $\left(\Theta_{n}\right)$ satisfies the recursive equation:
\[
\left\{ \begin{array}{cclc}
\Theta_{n}\left(q,\zeta\right) & = & \ell(q) & \text{if }n = N,\\
\Theta_{n}\left(q,\zeta\right) & = & \min\left\{ \tilde{\Theta}_{n,n+1}\left(q,\zeta\right),l(q)\right\}  & \text{if }n\in\mathcal{N},\\
\Theta_{n}\left(q,\zeta\right) & = & \tilde{\Theta}_{n,n+1}\left(q,\zeta\right) & \text{otherwise},
\end{array}\right.
\]
where
\begin{equation}
\begin{array}{l}
\tilde{\Theta}_{n,n+1}\left(q,\zeta\right)=\\
\begin{array}{ccl}
\inf_{\tilde{q}\in\R}\frac{1}{\gamma}\log\Bigg( & \sum_{j=0}^{4}p_{j}\exp\Bigg(\gamma\Bigg( & \sigma\sqrt{\Delta t}\left(\left(j-2\right)\left(q-\frac{q_{0}}{n+1}\right)-\frac{Q}{n+1}\left(\dfrac{\zeta}{n}-\left(n-1\right)\right)\right)\\
 &  & +L\left(\dfrac{q-\tilde{q}}{V_{n+1}\Delta t}\right)V_{n+1}\Delta t+\Theta_{n+1}\left(\tilde{q},\zeta+nj\right)\Bigg)\Bigg)\Bigg),
\end{array}
\end{array}\label{eq:Thetatilde-penta}
\end{equation}
with $\left(p_{j}\right)$ is given by:
\begin{align*}
\left\{ \begin{array}{ccccc}
p_{0} & = & p_{4} & = & \dfrac{1}{12}\\
p_{1} & = & p_{3} & = & \dfrac{1}{6}\\
p_{2} & = & \dfrac{1}{2}
\end{array}\right.
\end{align*}
\end{prop}

\begin{proof}
We first define $ \Theta_n $ for $ 1 \le n \le N$ as:
$$ \Theta_n(q,\zeta)= \theta_n\Big(q,\dfrac{\zeta}{n}-\left(n-1\right) \Big).$$
Using then the Bellman equation (\ref{eq:bellman-theta}) satisfied by $ (\theta_n)_{0\le n \le N}$, we have:
\[
\left\{ \begin{array}{cclc}
\Theta_{n}\left(q,\zeta\right) & = & \ell(q) & \text{if }n=N,\\
\Theta_{n}\left(q,\zeta\right) & = & \min\left\{ \tilde{\theta}_{n,n+1}\left(q,\dfrac{\zeta}{n}-\left(n-1\right)\right),l(q)\right\}  & \text{if }n\in\mathcal{N},\\
\Theta_{n}\left(q,\zeta\right) & = & \tilde{\theta}_{n,n+1}\left(q,\dfrac{\zeta}{n}-\left(n-1\right)\right) & \text{otherwise}.
\end{array}\right.
\]
We now use the definition of $ \tilde{\theta}_{n,n+1}$ to compute $ \tilde{\theta}_{n,n+1}\left(q,\dfrac{\zeta}{n}-\left(n-1\right) \right)$:
\begin{eqnarray*}
 &  & \tilde{\theta}_{n,n+1}\left(q,\dfrac{\zeta}{n}-\left(n-1\right) \right)\\
 & = & \inf_{\tilde{q}\in\R}\logexpo\sigma\sqrt{\delta t}\left(\left(q-\dfrac{Q}{n+1}\right)\epsilon_{n+1}-\dfrac{Q}{n+1}\left(\dfrac{\zeta}{n}-\left(n-1\right)\right)\right)\\
 &  & +L\left(\frac{q-\tilde{q}}{V_{n+1}\delta t}\right)V_{n+1}\delta t+\theta_{n+1}\left(\tilde{q},\dfrac{n}{n+1}\left(\left(\dfrac{\zeta}{n}-\left(n-1\right)\right)+\epsilon_{n+1}\right)\right)\logexpc.
\end{eqnarray*}
Noticing that $$ \theta_{n+1}\left(\tilde{q},\dfrac{n}{n+1}\left(\left(\dfrac{\zeta}{n}-\left(n-1\right)\right)+\epsilon_{n+1}\right)\right)=\Theta_{n+1}\left( \tilde{q},\zeta+n \big(\epsilon_{n+1}+2 \big)\right),$$ we obtain:
\begin{equation*}
\tilde{\theta}_{n,n+1}\left(q,\dfrac{\zeta}{n}-\left(n-1\right) \right)= \tilde{\Theta}_{n,n+1}\left(q,\zeta\right).
\end{equation*}
This is exactly the result.
\end{proof}

The new variable $\zeta$ is the index for nodes. To compute $\Pi$, we need to compute $\Theta_1(\cdot,0)$ which is computed using $\Theta_2(\cdot,0)$, $\Theta_2(\cdot,1)$, $\Theta_2(\cdot,2)$, $\Theta_2(\cdot,3)$ and $\Theta_2(\cdot,4)$. By induction, we see that, at step $n$, we shall need the values of  $\Theta_n(\cdot,0), \Theta_n(\cdot,1), \ldots, \Theta_n(\cdot,2n(n-1))$. In particular, the number of nodes at each level of the tree evolves in a quadratic way.\\
At each node $(n,\zeta)$ in the tree, we compute, using classical optimization methods, the values $\Theta_n(0,\zeta), \Theta_n(\delta q,\zeta), \ldots, \Theta_n(Q,\zeta)$, where $\delta q$ is the step between two consecutive values of $q$ in our grid. In addition to these values, we also set a table of boolean flags at each node $(\zeta,n)$ with $n\in \mathcal{N}$ to identify whether the bank should deliver the shares.

\subsection{Examples}

\subsubsection{Reference scenarios}

We now turn to the practical use of the above tree method. We consider the following reference case with no permanent market impact, that corresponds to rounded values for the stock Total SA. This case will be used throughout the remainder of this text.

\begin{itemize}
\item $S_{0}=45$ €
\item $\sigma=0.6$ €$\cdot\text{day}^{-1/2}$, which corresponds to an
annual volatility approximately equal to $21\%$.
\item $T=63$ trading days. The set of possible dates for delivery before expiry is $\mathcal{N} = [22,62]\cap \mathbb N$.
\item $V=4\ 000\ 000$ stocks$\cdot$ $\text{day}^{-1}$
\item $Q=20\ 000\ 000$ stocks
\item $L(\rho)=\eta|\rho|^{1+\phi}$ with $\eta=0.1$ € $\cdot\mbox{stock}^{-1}\cdot\text{day}^{-1}$
and $\phi=0.75$.
\end{itemize}
Furthermore, we force the bank to have already purchased all the stocks
at delivery. Therefore, we use $\ell(q)  =  +\infty 1_{\{q\neq0\}}.$\\

Our choice for risk aversion is $\gamma=2.5\times 10^{-7}$ €$^{-1}$.

To exemplify the model, we consider three trajectories for the price. These trajectories have been selected among a vast number of draws in order to exhibit several properties of the optimal strategy of the bank.

The first price trajectory (Figure \ref{price_up}) has an upward trend and the stock price is therefore above its average (dotted line). This corresponds to a positive $Z$ and the bank has no reason to deliver the shares rapidly. In this case indeed, the optimal strategy of the bank consists in buying the shares slowly to minimize execution costs, as exhibited on Figure \ref{ref_up}.

\begin{figure}[H]
\centering{}\includegraphics[width=0.68\textwidth]{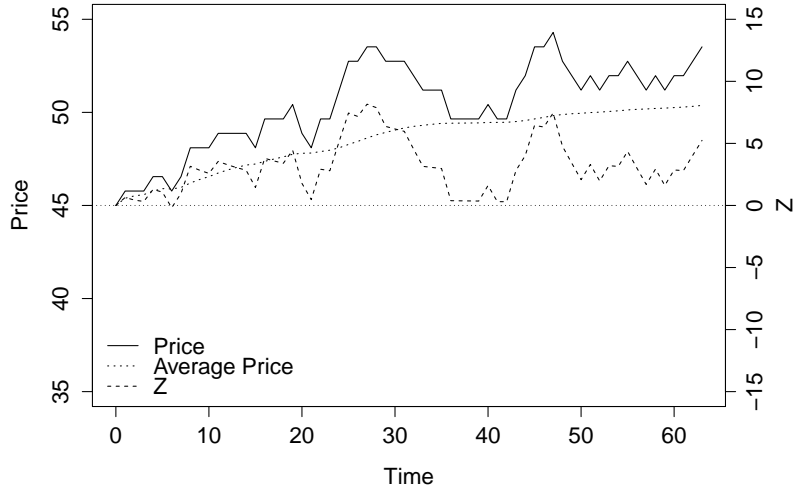}\caption{Price Trajectory 1}
\label{price_up}
\end{figure}

\begin{figure}[H]
\centering{}\includegraphics[width=0.68\textwidth]{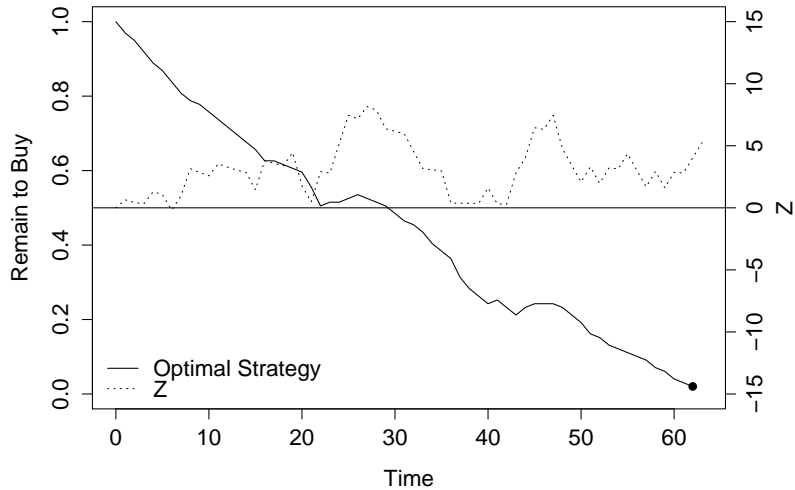}\caption{Optimal Strategy for Price Trajectory 1}
\label{ref_up}
\end{figure}

The second price trajectory (Figure \ref{price_down}) has a downward trend and the stock price is therefore below its average. This corresponds to a negative $Z$ and the bank has an incentive to deliver rapidly. However, the values of $Z$ are not low enough to encourage delivery when $n=22$ (the first date of $\mathcal{N}$). The trajectory of $Z$ and the level of execution costs eventually lead to delivery when $n=36$ -- see Figure \ref{ref_down}. We also see that it may even be optimal to sell shares when $Z$ increases after a decrease: this is in line with the risk term in equation (\ref{effects}) since an increase in $Z$ postpones the targeted time of delivery $t_{n^{\star}}$.

\begin{figure}[H]
\centering{}\includegraphics[width=0.7\textwidth]{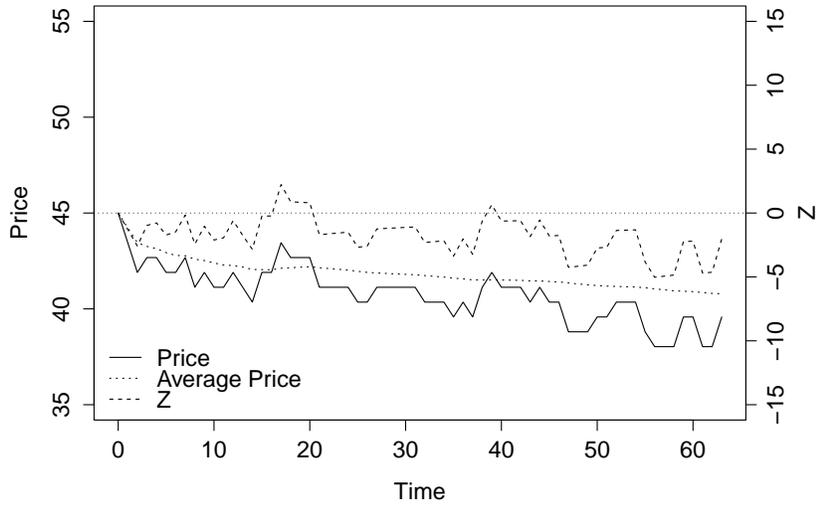}\caption{Price Trajectory 2}
\label{price_down}
\end{figure}

\begin{figure}[H]
\centering{}\includegraphics[width=0.7\textwidth]{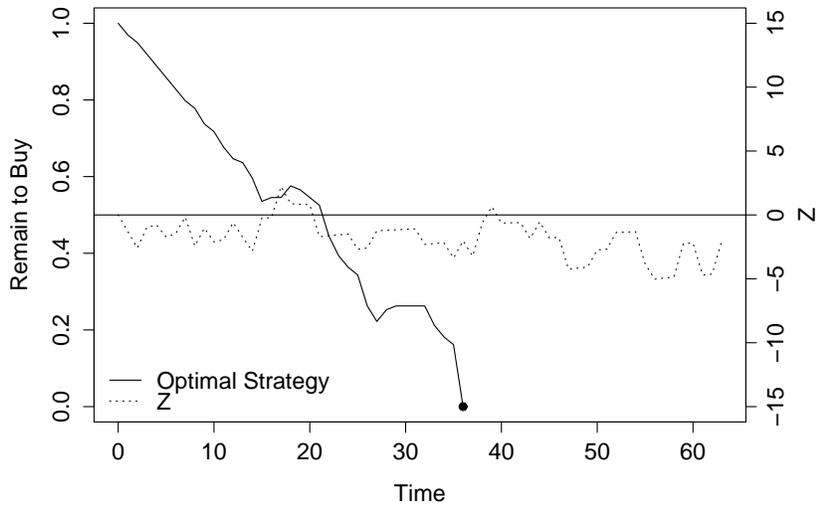}\caption{Optimal Strategy for Price Trajectory 2}
\label{ref_down}
\end{figure}

The third price trajectory we consider (Figure \ref{price_mid}) corresponds to $Z$ oscillating around $0$. As in the preceding two examples, we see on Figure \ref{ref_mid} that the behavior of the bank is strongly linked to $Z$ in a natural way. When $Z$ decreases to reach negative values, $A$ is larger than the cost of buying the shares (if execution costs are not too large) and the bank accelerates the buying process. On the contrary, when $Z$ increases, the buying process slows down or even turns  into a selling process (see day 16 on Figure \ref{ref_mid}). As above, the rationale for that is risk aversion: when $Z$ goes from a low value to a high value, the targeted time of delivery is somehow postponed and the bank has an incentive to jump to a trajectory that permits to hedge the payoff, as explained at the end of Section \ref{sect:change of var}. But the most interesting phenomenon appears after day 28: the excursion of $Z$ below $0$ made it optimal for the bank to hold a portfolio with $Q$ shares on day 28, that is all the shares needed to deliver... but delivery did not occur then. The bank was indeed long of a Bermudan option (in fact American here) with a complex payoff and did prefer to hold it instead of exercising it. The bank eventually delivered the shares at terminal time. The round trip on the underlying after the $28^{th}$ day corresponds then to mitigating the risk of the option.

\begin{figure}[H]
\centering{}\includegraphics[width=0.7\textwidth]{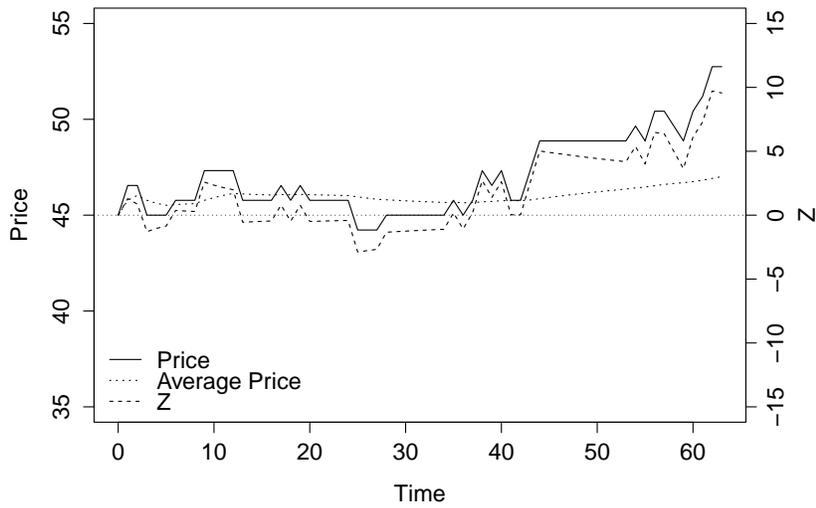}\caption{Price Trajectory 3}
\label{price_mid}
\end{figure}

\begin{figure}[H]
\centering{}\includegraphics[width=0.7\textwidth]{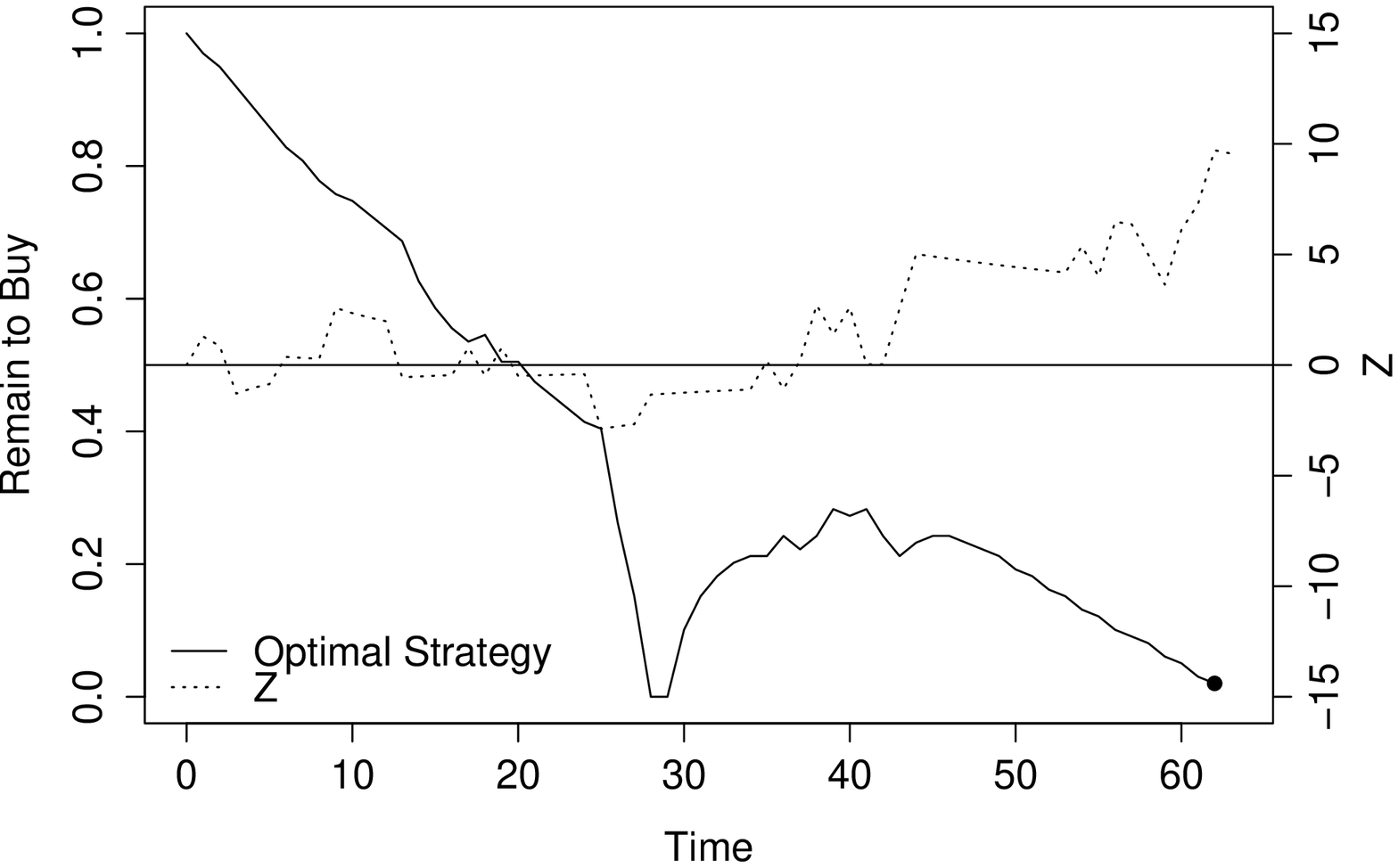}\caption{Optimal Strategy for Price Trajectory 3}
\label{ref_mid}
\end{figure}

In addition to optimal strategies, we can compute in our reference case, the price $\Pi$ of the ASR contract. Here, this price is negative as we found $\frac{\Pi}{Q} = -0.503$.\footnote{This does not mean that the price is linear in $Q$. It just means that, for the value of $Q$ we considered, we found $\frac{\Pi}{Q} = -0.503$.} This means, in utility terms, that the gain associated to the optionality component of the ASR contract is important enough to compensate (in utility terms) the liquidity costs and the risk of the contract (see Remark \ref{oboidorman}).

\subsubsection{Buy-Only Strategies}

We have seen above that the optimal strategy of the bank may involve selling the shares that have been bought. At first sight, this may look like arbitrage but it is not: it is rather a natural consequence of both the payoff of the ASR contract and the risk aversion of the bank. We exhibit on Figures \ref{constrained_up}, \ref{constrained_down} and \ref{constrained_mid}, for the three reference trajectories considered above, what would be the optimal strategy, had we restricted the admissible set of strategies to buy-only ones.

\begin{figure}[H]
\centering{}\includegraphics[width=0.7\textwidth]{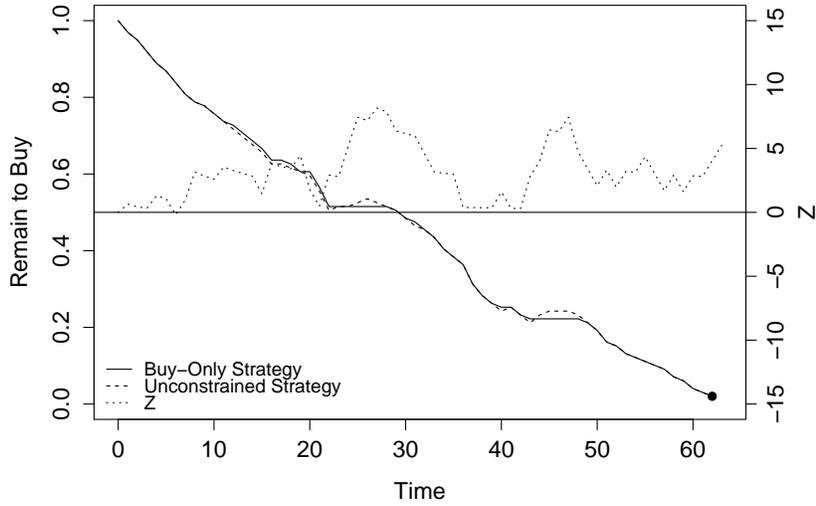}\caption{Buy-Only Strategy vs Unconstrained Strategy for Price Trajectory 1}
\label{constrained_up}
\end{figure}
\begin{figure}[H]
\centering{}\includegraphics[width=0.7\textwidth]{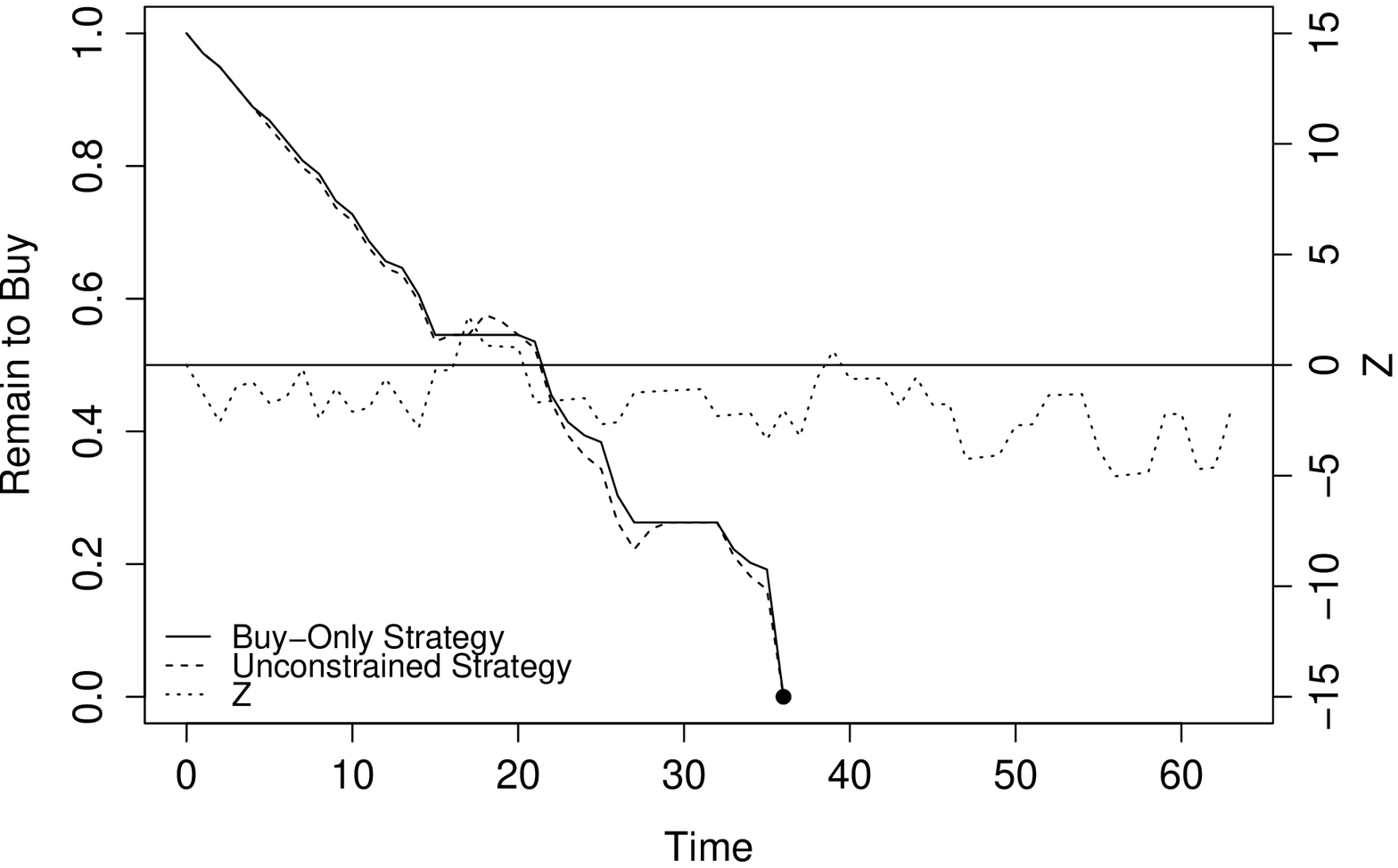}\caption{Buy-Only Strategy vs Unconstrained Strategy for Price Trajectory 2}
\label{constrained_down}
\end{figure}
\begin{figure}[H]
\begin{centering}
\includegraphics[width=0.7\textwidth]{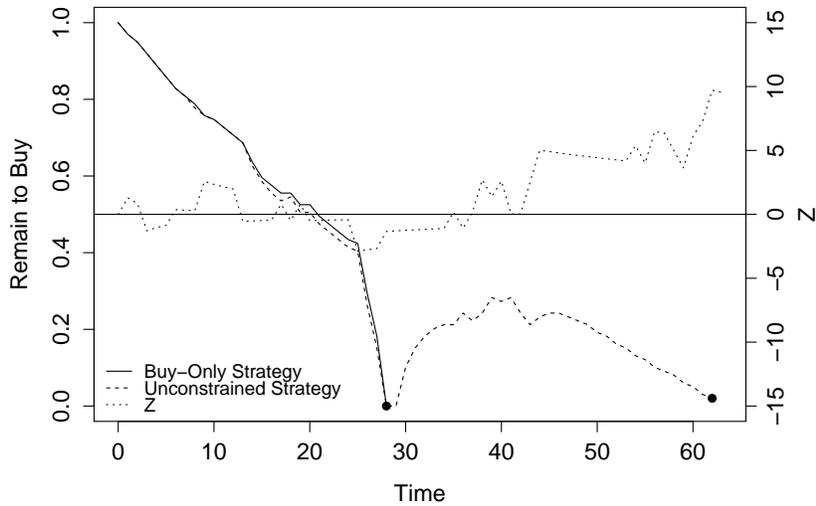}
\par\end{centering}

\centering{}\caption{Buy-Only Strategy vs Unconstrained Strategy for Price Trajectory 3}
\label{constrained_mid}
\end{figure}

We see on Figures \ref{constrained_up}, \ref{constrained_down} and \ref{constrained_mid} that, when round trips are not allowed, the bank usually buys more slowly in order to avoid the need of round trips. Another difference is magnified in the case of trajectory 3 as the bank delivers the shares more rapidly. The risk associated to the option cannot indeed be partially hedged as we restrict strategies to be buy-only strategies. The bank then decides to deliver when $n=28$.\\

To quantify the difference between the constrained case and the unconstrained case, we can compare the previous price $\Pi$ with the price $\Pi_{\text{constrained}}$ obtained when the set of strategies is limited to buy-only strategies. Not surprisingly $\frac{\Pi_{\text{constrained}}}{Q} = -0.486$ is larger than $\frac{\Pi}{Q}$ and the difference is rather low for the risk aversion parameter $\gamma$ chosen.

\subsection{Comparative Statics}

We now use the pentanomial tree method in order to study the effects of the parameters on the optimal strategy and on the price of the ASR contract. More precisely, we focus on the risk aversion parameter $\gamma$, on the illiquidity parameter $\eta$ and on the volatility parameter $\sigma$.\footnote{The influence of the nominal $Q$ can be deduced from the influence of $\eta$ and $\gamma$. Similarly, any multiplicative change in the volume curve can be analyzed as a change in $\eta$. The influence of $\phi$ is not studied however, as $\phi$ does not differ much across stocks.}

\subsubsection{Effect of Risk Aversion}

Let us focus first on risk aversion. We considered our reference case with 4 values\footnote{The formulas obtained in Section \ref{sect:reduc-var} can be extended to the case $\gamma = 0$.} for the parameter $\gamma$: $0, 2.5\times10^{-9}, 2.5\times10^{-7},$ and $2.5\times10^{-6}$. Figures \ref{gamma_up}, \ref{gamma_down} and \ref{gamma_mid} show the influence of $\gamma$ on the optimal strategy for the three reference stock price trajectories introduced previously. We see that the more risk averse the bank, the closer to the straight line its strategy. This is natural as the straight line strategy is a way to hedge perfectly against the risk associated to the payoff. At the other end of the spectrum, when $\gamma =0$, the optimal strategy is far from the diagonal and what prevents the bank from buying instantaneously is just execution costs. An interesting point is also that, when $\gamma = 0$, the optimal strategy does not involve any round trip.\footnote{This is in fact straightforward by induction.} Finally, we see on Figure \ref{gamma_mid} that $\gamma$ high ($2.5\times10^{-6}$) also deters the bank from delaying delivery.

\begin{figure}[H]
\centering{}\includegraphics[width=0.65\textwidth]{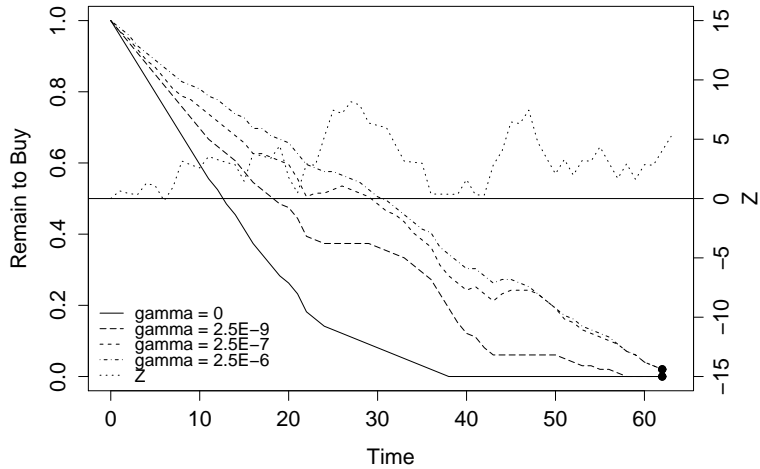}\caption{Optimal Strategies for Different Values of $\gamma$ for Price Trajectory
1}
\label{gamma_up}
\end{figure}
\begin{figure}[H]
\centering{}\includegraphics[width=0.65\textwidth]{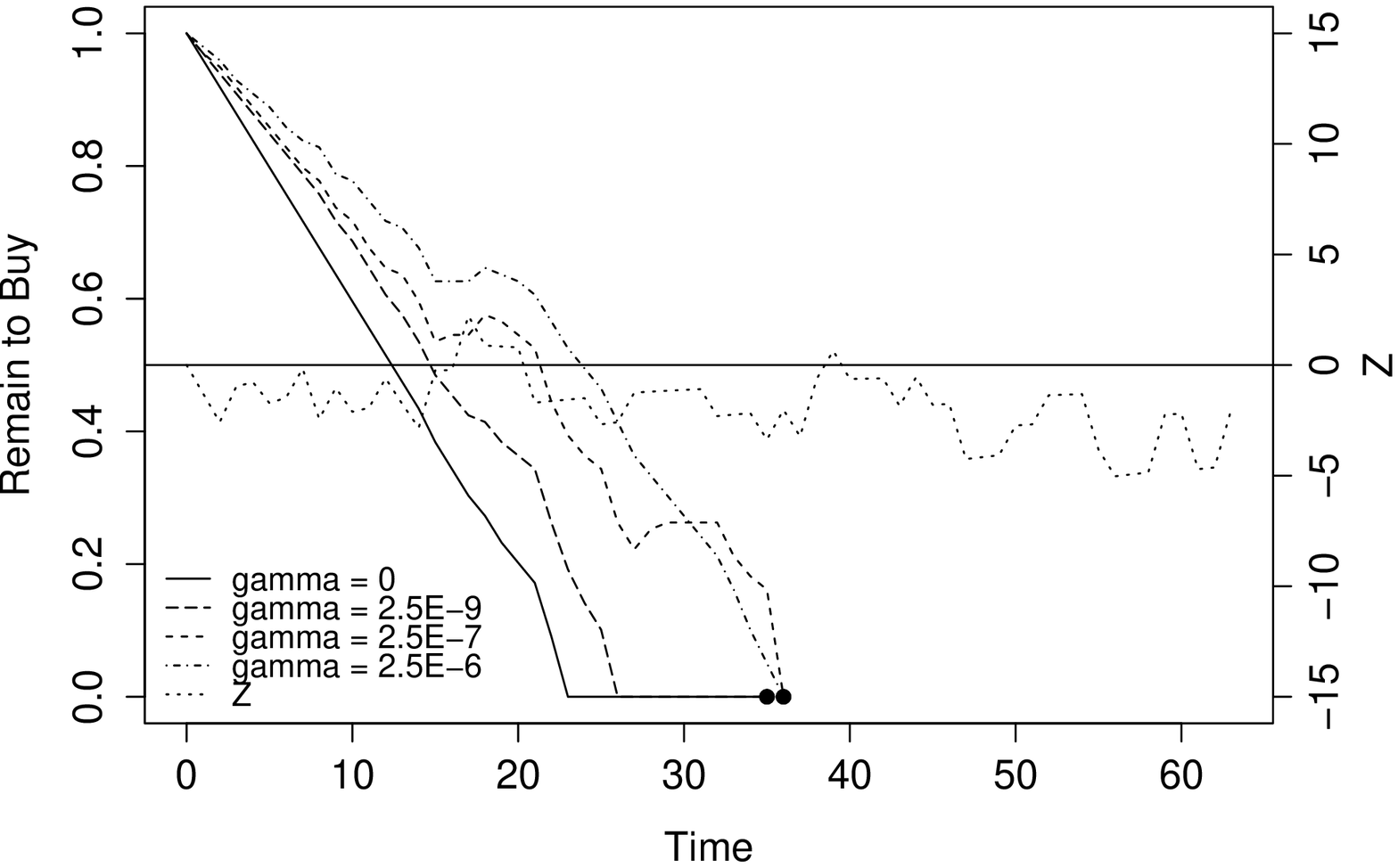}\caption{Optimal Strategies for Different Values of $\gamma$ for Price Trajectory
2}
\label{gamma_down}
\end{figure}
\begin{figure}[H]
\begin{centering}
\includegraphics[width=0.65\textwidth]{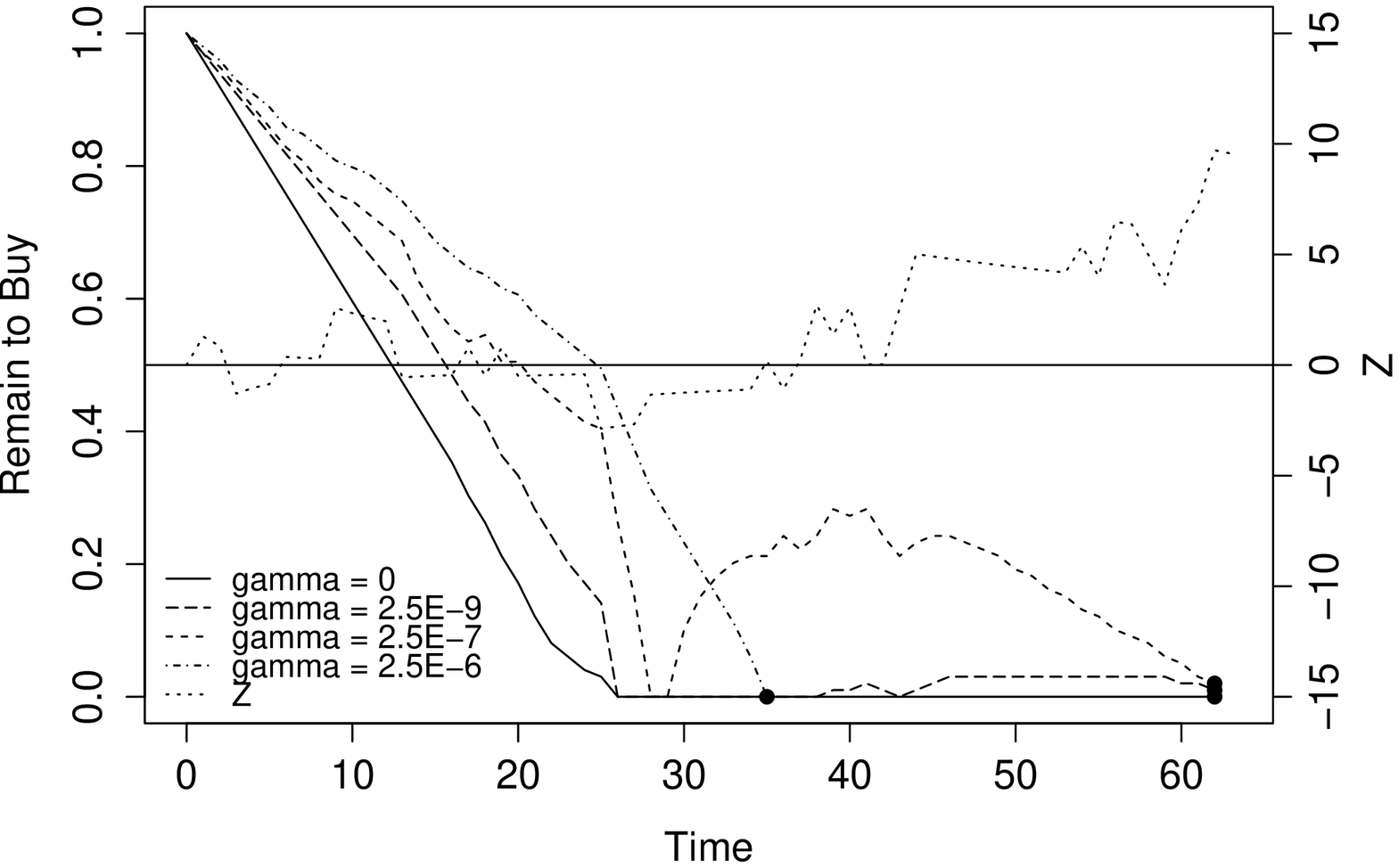}
\par\end{centering}

\centering{}\caption{Optimal Strategies for Different Values of $\gamma$ for Price Trajectory
3}
\label{gamma_mid}
\end{figure}

In terms of prices, we obtain:

\begin{table}[H]
\begin{centering}
\begin{tabular}{|c|c|c|c|c|}
\hline
$\gamma$ & $0$ & $2.5\times10^{-9}$ & $2.5\times10^{-7}$ & $2.5\times10^{-6}$\tabularnewline
\hline
$\frac{\Pi}{Q}$ & $-0.621$ & $-0.609$ & $-0.503$ & $-0.190$\tabularnewline
\hline
\end{tabular}
\par\end{centering}
\end{table}

We see that the price is an increasing function of the risk aversion parameter $\gamma$. This is natural as this price is an indifference price  that takes account of the risk. In particular, if we send $\gamma$ to unrealistically high values, the price turns out to be positive\footnote{For $\gamma = 1$ we obtained $\frac{\Pi}{Q} = 0.015$.}, meaning that the cost associated to risk and execution costs cannot be compensated by the value associated to the optionality of the ASR contract.

\subsubsection{Effect of Execution Costs}

Let us come then to execution costs and more precisely to the illiquidity parameter $\eta$. We considered our reference case with 3 values for the parameter $\eta$: $0.01, 0.1,$ and $0.2$. We concentrate on our third price trajectory (Figure \ref{eta_mid}) as it shows perfectly the role of $\eta$. We indeed see that the more liquid the stock, the faster $q$ reaches 0.

\begin{figure}[H]
\begin{centering}
\includegraphics[width=0.7\textwidth]{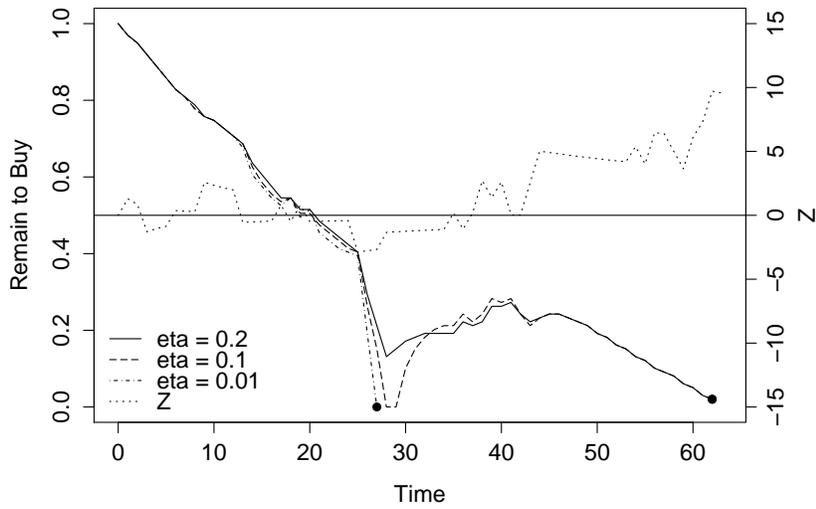}
\par\end{centering}

\centering{}\caption{Optimal Strategies for Different Values of $\eta$ for Price Trajectory
3}
\label{eta_mid}
\end{figure}

The influence of $\eta$ is important as far as prices are concerned since $\eta$ is the main driver of execution costs. We have indeed the following prices:

\begin{table}[H]
\begin{centering}
\begin{tabular}{|c|c|c|c|}
\hline
$\eta$ & $0.01$ & $0.1$ & $0.2$\tabularnewline
\hline
$\frac{\Pi}{Q}$ & -0.554  & $-0.503$ & $-0.461$\tabularnewline
\hline
\end{tabular}
\par\end{centering}
\end{table}

The less liquid the stock, the much it costs to the bank over the buying process.

\subsubsection{Effect of Volatility}

Let us finish with the effect of volatility. We considered our reference case with 3 values for the parameter $\sigma$: $0.3, 0.6,$ and $1.2$. We focus on the second price trajectory  as it permits to understand the role of volatility. What we see on Figure \ref{sigma_down} is that the higher the volatility the more important the size of the round trips. This is linked to risk aversion: the higher the volatility, the more important the incentive to jump on a trajectory corresponding to a better hedge after an increase in $Z$. However, $\sigma$ hardly influences the optimal delivery date as $A$ and $S$ are both influenced by $\sigma$.

\begin{figure}[H]
\centering{}\includegraphics[width=0.7\textwidth]{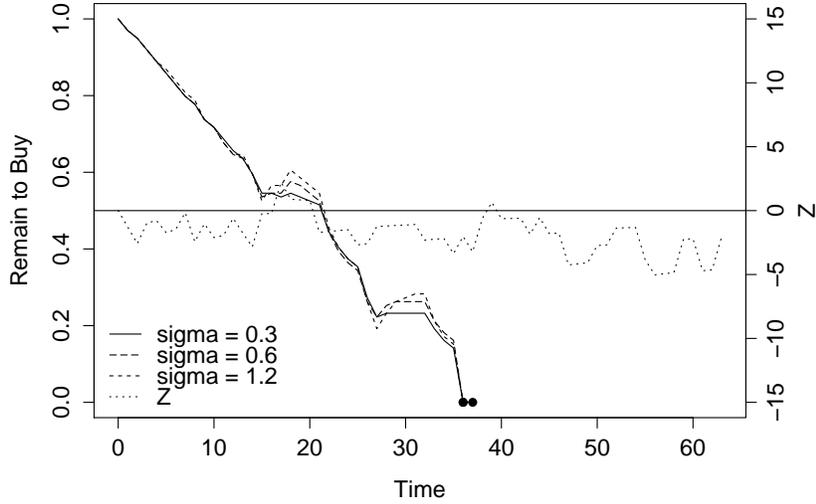}\caption{Optimal Strategies for Different Values of $\sigma$ for Price Trajectory
2}
\label{sigma_down}
\end{figure}

In terms of prices, we obtained the following figures. We see that $\sigma$ is a very important driver of prices.

\begin{table}[H]
\begin{centering}
\begin{tabular}{|c|c|c|c|}
\hline
$\sigma$ & $0.3$ & $0.6$ & $1.2$\tabularnewline
\hline
$\frac{\Pi}{Q}$ & $-0.251$ & $-0.503$ & $-0.914$\tabularnewline
\hline
\end{tabular}
\par\end{centering}

\end{table}

When $\sigma$ increases, the price decreases because the value of the optionality component is higher. However, due to risk aversion, the opposite effect may also be present. In fact, the latter effect only appeared to dominate the former for unrealistic values of $\sigma$ in our numerical simulations ($\sigma \ge 30$).

\section*{Conclusion}

In this paper, we presented a model to find the optimal strategy of a bank entering an ASR contract with a firm. Our discrete time model embeds the main effects involved in the problem and permits to define an indifference price for the contract in addition to providing an optimal strategy. One of the main limitations of our model is that only one decision is taken every day and hence the behavior of the bank is not influenced by intraday price changes. Taking account of intraday price changes requires considering a continuous model for $S$ along with price fixing at discrete times for $A$, making then impossible the reduction to a 3-dimension problem.

\section*{Appendix A: Bounds for $\left(\theta_{n}\right)_{n}$}
In the definition of $\left(\theta_{n}\right)_{n\geq0}$ of Section \ref{sect:reduc-var}, nothing
guarantees a priori that the functions only take finite values. We now provide
bounds in order to prove that the functions $\theta_{n}$ -- and therefore
$\Pi$ -- are finite.
\begin{prop}
$\forall n<N,q\in\R,Z\in\R,$
\begin{eqnarray*}
\theta_{n}\left(q,Z\right) & \leq & \frac{1}{\gamma}g\left(\gamma\left(q-Q\left(1-\dfrac{n}{N}\right)\right)\right)+\frac{1}{\gamma}\sum_{j=n+1}^{N-1}g\left(\gamma Q\left(1-\frac{j}{N}\right)\right)\\
 &  & -Q\left(1-\frac{n}{N}\right)\sigma\sqrt{\delta t}Z+L\left(\frac{q}{V_{n+1}\delta t}\right)V_{n+1}\delta t
\end{eqnarray*}
where $g$ denotes the cumulant-generating function of the random
variable $\sigma\sqrt{\delta t}\epsilon_{1}$.\end{prop}
\begin{proof}
Given $n<N$, $q$ and $Z\in\R$, we consider the strategy consisting
in sending an order of size $q$ at time $t_n$ and wait until
maturity $T$ to deliver the shares, which is given by:
\[
\left\{ \begin{array}{cclc}
v_{n} & = & q/\delta t\\
v_{j} & = & 0 & \forall j\in\{n+1,\ldots,N-1\}\\
n^{\star} & = & N
\end{array}\right.
\]
Then we have:
\begin{eqnarray*}
\theta_{n}\left(q,Z\right) & \leq & \logexpo\sigma\sqrt{\delta t}\left(q\epsilon_{n+1}-\sum_{j=n}^{N-1}\left(Q\left(1-\frac{j}{N}\right)\right)\epsilon_{j+1}-Q\left(1-\frac{n}{N}\right)Z\right)\\
 &  & +L\left(\frac{q}{V_{n+1}\delta t}\right)V_{n+1}\delta t\logexpc\\
 & \leq & \frac{1}{\gamma}g\left(\gamma\left(q-Q\left(1-\dfrac{n}{N}\right)\right)\right)+\frac{1}{\gamma}\sum_{j=n+1}^{N-1}g\left(\gamma Q\left(1-\frac{j}{N}\right)\right)\\
 &  & -Q\left(1-\frac{n}{N}\right)\sigma\sqrt{\delta t}Z+L\left(\frac{q}{V_{n+1}\delta t}\right)V_{n+1}\delta t
\end{eqnarray*}
where we used the fact that $\left(\epsilon_{j}\right)_j$ are i.i.d.
random variables.\end{proof}
\begin{prop}
We also provide a lower bound for $\theta_{n}$: $\forall n\le N,\exists C_{n},D_{n}\in\R_{+},$
\begin{eqnarray*}
\theta_{n}(q,Z) & \ge & -C_{n}Z^{+}-D_{n}
\end{eqnarray*}
\end{prop}
\begin{proof}
We prove this proposition by backward induction.

The proposition is true for $n=N$ with $C_{N}=0$ and $D_{N}=0$
as $\ell\ge0$.

We have for $n\in\{1,\ldots,N-1\}$,
\begin{eqnarray*}
\tilde{\theta}_{n,n+1}\left(q,Z\right) & = & \inf_{v\in\R}\logexpo\sigma\sqrt{\delta t}\left(\left(q-\dfrac{Q}{n+1}\right)\epsilon_{j+1}-\dfrac{Q}{n+1}Z\right)\\
 &  & +L\left(\frac{v}{V_{n+1}}\right)V_{n+1}\delta t+\theta_{n+1}\left(q-v\delta t,\dfrac{n}{n+1}\left(Z+\epsilon_{n+1}\right)\right)\logexpc\\
 & \geq & -\sigma\sqrt{\delta t}\dfrac{Q}{n+1}Z+\inf_{v\in\R}\mathbb{E}\Biggl[\theta_{n+1}\left(q-v\delta t,\dfrac{n}{n+1}\left(Z+\epsilon_{n+1}\right)\right)\Biggr].
\end{eqnarray*}
Assume we have $C_{n+1},D_{n+1}\in\R_{+}$ such that $\forall q,Z\in\R,$
\begin{eqnarray*}
\theta_{n+1}(q,Z) & \geq & -C_{n+1}Z^{+}-D_{n+1}.
\end{eqnarray*}
Then
\begin{eqnarray*}
\tilde{\theta}_{n,n+1}\left(q,Z\right) & \geq & -\sigma\sqrt{\delta t}\dfrac{Q}{n+1}Z+\mathbb{E}\Biggl[-C_{n+1}\dfrac{n}{n+1}\left(Z+\epsilon_{n+1}\right)^{+}-D_{n+1}\Biggr]\\
 & \geq & -\sigma\sqrt{\delta t}\dfrac{Q}{n+1}Z^{+}-C_{n+1}\dfrac{n}{n+1}\mathbb{E}\left[\left(Z+\epsilon_{n+1}\right)^{+}\right]-D_{n+1}\\
 & \geq & -\sigma\sqrt{\delta t}\dfrac{Q}{n+1}Z^{+}-C_{n+1}\dfrac{n}{n+1}\mathbb{E}\left[Z^{+}+\epsilon_{n+1}^{+}\right]-D_{n+1}\\
 & \geq & -\sigma\sqrt{\delta t}\dfrac{Q}{n+1}Z^{+}-C_{n+1}\dfrac{n}{n+1}Z^{+}-C_{n+1}\dfrac{n}{n+1}\mathbb{E}\left[\epsilon_{n+1}^{+}\right]-D_{n+1}\\
 & \geq & -\left(C_{n+1}\dfrac{n}{n+1}+\sigma\sqrt{\delta t}\dfrac{Q}{n+1}\right)Z^{+}-\left(C_{n+1}\dfrac{n}{n+1}\mathbb{E}\left[\epsilon_{n+1}^{+}\right]+D_{n+1}\right).
\end{eqnarray*}
Let us define
\[
\begin{array}{cclcc}
C_{n} & = & \frac{n}{n+1}C_{n+1}+\sigma\sqrt{\delta t}\dfrac{Q}{n+1} & \ge & 0\\
D_{n} & = & \frac{n}{n+1}C_{n+1}\mathbb{E}\left[\epsilon_{n+1}^{+}\right]+D_{n+1} & \ge & 0
\end{array}
\]
Then
\begin{align*}
\theta_{n}(q,Z) &=  \min\{\ell(q),\tilde{\theta}_{n,n+1}(q,Z)\}\\
 &\geq \min\{0,\tilde{\theta}_{n,n+1}(q,Z)\} \geq \min\{0,-C_{n}Z^{+}-D_{n}\} =  -C_{n}Z^{+}-D_{n}.
\end{align*}
\end{proof}

\section*{Appendix B: A model in discrete time for permanent market impact}

In Section \ref{sect: permanent market impact}, we introduced a model in discrete time for permanent market impact. This model comes from the following model in continuous time:
\begin{itemize}
\item The number of shares to be bought evolves as:
$$d\check{q}_{t}  =  -\check{v}_{t}dt.$$
\item The price has the following dynamics:
$$d\check{S}_{t}  =  \sigma dW_{t}+f(|Q-\check{q}_{t}|)\check{v}_{t}dt.$$
\item The cash account evolves as:
$$d\check{X}_{t}  =  \check{v}_{t}\check{S}_{t}dt+L\left(\dfrac{\check{v}_{t}}{\check{V}_{t}}\right)\check{V}_{t}dt.$$
\end{itemize}
Solving the above stochastic differential equations, we obtain $\forall s<t$:
\begin{align*}
\check{S}_{t}-\check{S}_{s} &= \sigma\left(W_{t}-W_{s}\right)+\int_{s}^{t}f(|Q-\check{q}_{r}|)\check{v}_{r}dr\\
 &= \sigma\left(W_{t}-W_{s}\right)-\int_{q_{s}}^{q_{t}}f(|Q-y|)dy\\
 &= \sigma\left(W_{t}-W_{s}\right)+G(q_{t})-G(q_{s}).
\end{align*}
Writing $S_n = \check{S}_{t_n}$, we obtain:
$$ S_{n+1}=  S_{n}+\sigma\sqrt{\delta t}\epsilon_{n+1}+\left(G(q_{n+1})-G(q_{n})\right), $$
where $\epsilon_{n+1} = \frac{W_{t_{n+1}}-W_{t_n}}{\sqrt{\delta t}}$ is a standard normal random variable. This is in line with the setup of the model of Section \ref{sect: Optimal strategy}.

Coming now to the cash account we have $\forall s<t$:
$$ \check{X}_{t}-\check{X}_{s} =  \int_{s}^{t}\check{v}_{r}\check{S}_{r}dr+\int_{s}^{t}L\left(\dfrac{\check{v}_{r}}{\check{V}_{r}}\right)\check{V}_{r}dr. $$
After integrating by parts, we obtain:
\begin{align*}
\check{X}_{t}-\check{X}_{s} = &\int_{s}^{t}\check{v}_{r}\check{S}_{r}dr+\int_{s}^{t}L\left(\dfrac{\check{v}_{r}}{\check{V}_{r}}\right)\check{V}_{r}dr\\
= &(\check{q}_{s}-\check{q}_{t})\check{S}_{t}-\int_{s}^{t}(\check{q}_{s}-\check{q}_{r})\sigma dW_{r}-\int_{s}^{t}(\check{q}_{s}-\check{q}_{r})f(|Q-\check{q}_{r}|)\check{v}_{r}dr+\int_{s}^{t}L\left(\dfrac{\check{v}_{r}}{\check{V}_{r}}\right)\check{V}_{r}dr\\
= &(\check{q}_{s}-\check{q}_{t})\check{S}_{t}-\int_{s}^{t}(\check{q}_{s}-\check{q}_{r})\sigma dW_{r}+\int_{\check{q}_{s}}^{\check{q}_{t}}(\check{q}_{s}-y)f(|Q-y|)dy+\int_{s}^{t}L\left(\dfrac{\check{v}_{r}}{\check{V}_{r}}\right)\check{V}_{r}dr\\
= &(\check{q}_{s}-\check{q}_{t})\check{S}_{t}-\int_{s}^{t}(\check{q}_{s}-\check{q}_{r})\sigma dW_{r}+\check{q}_{s}\int_{\check{q}_{s}}^{\check{q}_{t}}f(|Q-y|)dy\\
& -\int_{\check{q}_{s}}^{\check{q}_{t}}yf(|Q-y|)dy+\int_{s}^{t}L\left(\dfrac{\check{v}_{r}}{\check{V}_{r}}\right)\check{V}_{r}dr\\
= &(\check{q}_{s}-\check{q}_{t})\check{S}_{t}-\int_{s}^{t}(\check{q}_{s}-\check{q}_{r})\sigma dW_{r}-\check{q}_{s}\left(G(\check{q}_{t})-G(\check{q}_{s})\right)\\
&+\left(F(\check{q}_{t})-F(\check{q}_{s})\right)+\int_{s}^{t}L\left(\dfrac{\check{v}_{r}}{\check{V}_{r}}\right)\check{V}_{r}dr.
\end{align*}
Now, writing $S_n = \check{S}_{t_n}$, $X_n =  \check{X}_{t_n}$ and assuming that $\check{V}_t$ and $\check{v}_t$ are both constant on $[t_n, t_{n+1}]$, respectively equal to $V_{n+1}$ and $v_n$, we obtain:
\begin{align*}
X_{n+1}-X_{n} =& S_{n+1}v_{n}\delta t-\sigma v_{n}\frac{\delta t^{\frac{3}{2}}}{\sqrt{3}}\epsilon'_{n+1}-q_{n}\left(G(q_{n+1})-G(q_{n})\right)\\
& +\left(F(q_{n+1})-F(q_{n})\right)+L\left(\dfrac{v_{n}}{V_{n+1}}\right)V_{n+1}\delta t,
\end{align*}
where $\epsilon'_{n+1} = \frac{\sqrt{3}}{\delta t^{\frac{3}{2}}}\int_{t_n}^{t_{n+1}}\left(r-t_n\right)dW_{r}$ is such that $\left(\epsilon_{n+1},\epsilon'_{n+1}\right)  \sim \mathcal{N}\left(0,\left(\begin{array}{cc}
1 & \dfrac{\sqrt{3}}{2}\\
\dfrac{\sqrt{3}}{2} & 1
\end{array}\right)\right).$ This is in line with the setup of the model of Section \ref{sect: permanent market impact}.

\vspace{3mm}
We finish this appendix with a focus on the term involved in the final cost. If we consider two times $s$ and $t$ with $s<t$ and $\check{q}_t=0$, then:
\begin{align*}
\check{X}_{t}-\check{X}_{s} &=\int_{s}^{t}\check{v}_{r}\check{S}_{r}dr+\int_{s}^{t}L\left(\dfrac{\check{v}_{r}}{\check{V}_{r}}\right)\check{V}_{r}dr\\
 &=\check{q}_{s}\check{S}_{s}+\int_{s}^{t}\check{q}_{r}\sigma dW_{r}+\int_{s}^{t}\check{q}_{r}f(|Q-\check{q}_{r}|)\check{v}_{r}dr+\int_{s}^{t}L\left(\dfrac{\check{v}_{r}}{\check{V}_{r}}\right)\check{V}_{r}dr\\
  &=\check{q}_{s}\check{S}_{s}+\int_{s}^{t}\check{q}_{r}\sigma dW_{r}+ F(0)-F(\check{q}_{s})+\int_{s}^{t}L\left(\dfrac{\check{v}_{r}}{\check{V}_{r}}\right)\check{V}_{r}dr.
\end{align*}
If $s$ corresponds to the time of delivery, then we can consider that we buy the shares remaining to be bought (that is $q_{n^\star}$ in our model) in exchange of $q_{n^\star}S_{n^\star} + F(0) - F(q_{n^\star})$ plus a risk liquidity premium represented by $\ell(q_{n^\star})$.

\bibliographystyle{plain}

\end{document}